\documentclass[lettersize,journal]{IEEEtran}
\usepackage{amsmath,amsfonts}
\usepackage{amsthm}
\usepackage{breqn}
\usepackage{algorithm}
\usepackage{algpseudocode}
\usepackage{xcolor}
\usepackage{array}
\usepackage[caption=false,font=normalsize,labelfont=sf,textfont=sf]{subfig}
\usepackage{textcomp}
\usepackage{stfloats}
\usepackage{url}
\usepackage{verbatim}
\usepackage{graphicx}
\usepackage{cite}
\usepackage[hidelinks]{hyperref}
\usepackage[acronym]{glossaries-extra}
\usepackage[capitalize]{cleveref}
\setabbreviationstyle[acronym]{long-short}
\glssetcategoryattribute{acronym}{nohyperfirst}{true}

\newtheorem{theorem}{Theorem}
\newtheorem{lemma}{Lemma}

\theoremstyle{definition}
\newtheorem{definition}{Definition}
\newtheorem{assumption}{Assumption}
\newtheorem{formulation}{Formulation}
\newtheorem{problem}{Problem}

\theoremstyle{remark}
\newtheorem{remark}{Remark}

\renewcommand{\baselinestretch}{0.968}

\newacronym{acr:amod}{AMoD}{Autonomous Mobility-on-Demand}
\newacronym{acr:mod}{MoD}{Mobility-on-Demand}
\newacronym{acr:od}{OD}{Origin-Destination}
\newacronym{acr:odd}{ODD}{Operational Design Domain}
\newacronym{acr:andp}{AMoD-NDP}{Autonomous Mobility-on-Demand Network Design Problem}
\newacronym{acr:ndp}{TNDP}{Transportation Network Design Problem}
\newacronym{acr:cg}{CG}{Column Generation}
\newacronym{acr:mp}{MP}{Master Problem}
\newacronym{acr:rmp}{RMP}{Restricted Master Problem}
\newacronym{acr:milp}{MILP}{Mixed Integer Liner Program}
\newacronym{acr:lp}{LP}{Linear Program}
\newacronym{acr:sprc}{SPRC}{Shortest Path with Resource Constraint}
\newacronym{acr:ro}{RO}{Robust Optimization}
\newacronym{acr:av}{AV}{Autonomous Vehicle}
\newacronym{acr:qos}{QoS}{Quality of Service}

\newacronym{acr:ue}{UE}{User Equilibrium}
\newacronym{acr:sue}{SUE}{Stochastic User Equilibrium}

\newcommand{\graph}{\mathcal{G}}
\newcommand{\nodes}{\mathcal{V}}
\newcommand{\edges}{\mathcal{E}}
\newcommand{\tup}[1]{\left(#1\right)}
\newcommand{\des}{\mathbf{x}}
\newcommand{\flow}{\mathbf{y}}

\begin{document}

\title{Where Should Robotaxis Operate? \\Strategic Network Design for Autonomous Mobility-on-Demand}

\author{Xinling Li, Gioele Zardini
\thanks{The authors are with the Laboratory for Information and Decision Systems, Massachusetts Institute of Technology, Cambridge, MA 02139, USA, {\tt \{xinli831,gzardini\}@mit.edu}}
\thanks{This work was supported by the Sidara Urban Seed Grant Program at the Norman B. Leventhal Center for Advanced Urbanism, Massachusetts Institute of Technology, and by Xinling Li's Mathworks Fellowship.}}



\maketitle

\begin{abstract}
The emergence of \gls{acr:amod} services creates new opportunities to improve the efficiency and reliability of on-demand mobility systems. 
Unlike human-driven \gls{acr:mod}, \gls{acr:amod} enables fully centralized fleet control, but it also requires appropriate infrastructure, so that vehicles can operate safely only on a suitably instrumented subnetwork of the roads. 
Most existing \gls{acr:amod} research focuses on fleet control (matching, rebalancing, ridepooling) on a fixed road network and does not address the joint design of the service network and fleet capacity.
In this paper, we formalize this strategic design problem as the \gls{acr:andp}, in which an operator selects an operation subnetwork and routes all passengers, subject to infrastructure and fleet constraints and route-level quality-of-service requirements. 
We propose a path-based mixed-integer formulation of the \gls{acr:andp} and develop a column-generation-based algorithm that scales to city-sized networks. 
The master problem optimizes over a restricted set of paths, while the pricing problem reduces to an elementary shortest path with resource constraints, solved exactly by a tailored label-correcting algorithm. 
The method provides an explicit certificate of the optimality gap and extends naturally to a robust counterpart under box uncertainty in travel times and demand.
Using real-world data from Manhattan, New York City, we show that the framework produces stable and interpretable operation subnetworks, quantifies trade-offs between infrastructure investment and fleet time, and accommodates additional path-level constraints, such as limits on left turns as a proxy for operational risk. 
These results illustrate how the proposed approach can support strategic planning and policy analysis for future \gls{acr:amod} deployments.
\end{abstract}


\section{Introduction}\label{sec:intro}
Rapid urbanization continues to increase transportation demand in major cities worldwide~\cite{un_climate_cities_pollution}. 
The resulting growth in private vehicle use has led to severe congestion, with substantial losses in productive time and increased emissions~\cite{schrank20112011}. 
According to the Intergovernmental Panel on Climate Change, the transport sector alone accounts for roughly 15\% of global greenhouse gas emissions~\cite{ipcc_ar6_wg3_transport}.
Addressing these challenges requires not only new vehicle technologies, but also new ways of organizing and operating urban mobility systems.

Over the past decade, \gls{acr:mod} services such as Uber and Lyft have emerged as flexible, on-demand alternatives to private car ownership, and have experienced rapid global growth~\cite{marketsandmarkets}.
While \gls{acr:mod} can improve accessibility and reduce parking demand, empirical evidence indicates that, in their current form, these services can also exacerbate congestion and emissions due to high levels of empty vehicle travel and deadheading~\cite{beojone2021inefficiency,erhardt2019transportation}.
To realize the potential benefits of \gls{acr:mod}, vehicles must be coordinated efficiently so that demand is served with minimal cruising and spatial mismatch between supply and demand.
In practice, achieving such coordination is challenging in human-driven \gls{acr:mod} systems because drivers act autonomously and may not follow system-optimal repositioning or routing policies.
Recent studies show that decentralized driver behavior and self-repositioning can lead to inefficient dispatch patterns and excessive cruising, even when the platform uses sophisticated matching algorithms~\cite{castillo2025matching, buchholz2015spatial, brar2024vehicle}.

The recent emergence of \gls{acr:mod} with \glspl{acr:av}, or \gls{acr:amod}, fundamentally changes this landscape.
In \gls{acr:amod} systems, vehicles are owned and operated by a central operator, enabling fully centralized fleet coordination.
A growing body of work has therefore focused on operational control of \gls{acr:amod}, including matching, rebalancing, and ridepooling algorithms~\cite{zardini2022analysis,li2025reproducibility}.
However, most existing models assume that the fleet operates on a \emph{fixed} underlying road network, following the standard assumptions inherited from human-driven \gls{acr:mod}.
This assumption is reasonable when drivers supply their own vehicles and operate on the existing network, which is outside of the control of the platform.

For \gls{acr:amod}, by contrast, both the fleet and the service network are under the operator's control. 
\glspl{acr:av} require appropriate physical and/or digital infrastructure for safe operation, and most deployments are initially restricted to a limited \gls{acr:odd}~\cite{saeed2019road,kockelman2017assessment}. 
The operator must therefore decide \emph{which} parts of the road network to instrument or leverage for autonomous operation and \emph{how large} a fleet to deploy on the resulting operation network.
These long-term design choices directly affect service profitability, operational risk, and system performance.
They shape feasible routes, travel times, and the reliability of fleet operations.
In this sense, \gls{acr:amod} poses a new type of \gls{acr:ndp}, in which a centralized operator jointly designs the service network and routes vehicles and passengers on it.

This paper addresses the resulting strategic design question:
\emph{how should an \gls{acr:amod} operator jointly design its service
network and fleet capacity, subject to infrastructure budgets,
fleet-time limits, and route-level quality constraints?}
Answering this question requires moving beyond fixed-network
models and explicitly incorporating network instrumentation
decisions into the \gls{acr:amod} optimization framework.

\paragraph*{Statement of contribution}
In this work, we make three main contributions.
First, we formally define the \gls{acr:andp}, in which a centralized operator selects a subnetwork of the road graph for autonomous operation and jointly determines passenger routing and fleet utilization.
The formulation captures infrastructure investment, fleet-time constraints (directly linked to fleet size), and route-level \gls{acr:qos} and risk constraints.
Second, we develop a scalable path-based mixed-integer formulation and a \gls{acr:cg}-based algorithm tailored to the \gls{acr:andp}.
The \gls{acr:lp} master problem optimizes over a restricted set of paths, and the pricing problem reduces to an elementary \gls{acr:sprc} on the road network.
We design an exact label-correcting pricing algorithm with \gls{acr:od}-specific preprocessing and prove that the overall method yields the optimal \gls{acr:lp} relaxation value together with an explicit optimality gap certificate for the recovered integer solution.
We further show that, under box uncertainty in travel times and demand, the robust counterpart preserves the same decomposition structure.
Finally, using real-world data from Mahattan, NYC,  we demonstrate that the framework scales to city-sized instances and produces interpretable operation subnetworks.
Through case studies, we quantify trade-offs between infrastructure budgets and fleet size and illustrate how additional path-level constraints, such as limits on left turns as a proxy for operational risk, affect profitability and service coverage.
These experiments highlight how the proposed model can serve as a decision-support tool for both operators and municipalities.

\paragraph*{Organization of the manuscript}
The remainder of this paper is organized as follows.
In \Cref{sec:review}, we review the literature on network design for urban transportation systems and position the \gls{acr:andp} relative to road and transit network design.
In \Cref{sec:system-model}, we formally define the \gls{acr:andp} and present the model formulation.
In \Cref{sec:methodology}, we introduce the \gls{acr:cg}-based algorithm and its robust extension.
Case studies are presented in \Cref{sec:exp} to illustrate the flexibility of the model and the effectiveness of the proposed algorithm.
Finally, \Cref{sec:conclusion} concludes with a discussion of future research directions.

\section{Related Work} 
\label{sec:review}
The \gls{acr:ndp} has been a central topic in transportation and infrastructure planning for decades. 
Broadly speaking, it seeks to design or improve a transportation network so as to optimize a system-level performance metric under demand and budget constraints.
Two classical streams of research are the \emph{road network design} problem and the \emph{transit network design} problem.
In this section we briefly review these two areas and then highlight the key differences that distinguish the proposed \gls{acr:andp} from both.

\subsection{Road network design}
The road network design problem focuses on modifying an existing road network to improve system performance, typically aiming to reduce congestion or average travel time.
Depending on the decision variables, it is usually classified into three types:
\begin{enumerate}
    \item \emph{Continuous} road network design, which adjusts the capacities of existing road segments.
    \item \emph{Discrete} road network design, which decides whether to add or remove individual links.
    \item \emph{Mixed} road network design, which combines both capacity changes and discrete link additions.
\end{enumerate}

Regardless of the type, road network design is mostly formulated as a bi-level optimization problem. 
On the upper level, a system regulator decides how to invest in the road infrastructure subject to a limited budget.
On the lower level, travelers react to these network changes by choosing routes according to some equilibrium principle, typically a \gls{acr:ue} or \gls{acr:sue}.

The notion of \gls{acr:ue} originates from Wardrop's first principle~\cite{wardrop1952road}, which describes a network game where drivers selfishly route themselves to minimize their own travel cost when that cost depends on link flows.
The equilibrium flow can be computed using Beckmann's transformation, which provides an equivalent single-level optimization formulation whose optimal solution satisfies the \gls{acr:ue} conditions~\cite{sheffi1985urban}.
This framework has been extended to account for randomness in perceived travel cost, leading to \gls{acr:sue} models~\cite{daganzo1977stochastic}.
Both \gls{acr:ue} and \gls{acr:sue} form the standard lower-level models in bi-level road network design.

With flows characterized by equilibrium conditions, the road network design problem becomes a bi-level optimization problem with equilibrium constraints, which is NP-hard in general.
Early work such as~\cite{leblanc1975algorithm} considered discrete network design with a \gls{acr:ue} lower level and proposed an exact branch-and-bound algorithm that explicitly accounts for phenomena such as Braess' paradox.
While exact, these methods do not scale to realistic urban networks.

Subsequent research has focused on improving scalability, leading to a rich toolbox including relaxation methods~\cite{poorzahedy1982approximate,wang2013global}, heuristics and metaheuristics~\cite{chen1991network,poorzahedy2007hybrid}, decomposition approaches~\cite{gao2005solution,fontaine2014benders}, and single-level reformulations based on variational inequalities or Karush-Kuhn-Tucker (KKT) conditions~\cite{farvaresh2011single,luathep2011global}.
Continuous road design, in which capacities rather than discrete links are optimized, has a similar structure and is often used as a tractable approximation of discrete problems~\cite{dantzig1979formulating,abdulaal1979continuous}.
However, even linear bi-level problems remain NP-hard~\cite{ben1990computational}, and a large body of work continues to develop heuristics and approximation schemes for the continuous case~\cite{marcotte1986network,abdulaal1979continuous,chiou2005bilevel}.

Comprehensive surveys such as~\cite{farahani2013review} review these models and algorithms in detail.
We emphasize that, in all of these works, responsibility for network design and route choice is split: a regulator designs the infrastructure, and self-interested drivers route themselves, leading to the characteristic bi-level structure~\cite{bard1982explicit}.
While in a broader context road network design can also encompasses the regulation of intersections, traffic signaling, and road tolling schemes, this is not the focus of this paper.

\subsection{Transit network design}
\label{subsec:related-transit}
Transit network design has a long history and has grown in complexity and realism over time.
In its basic form, the problem is to design a set of transit routes (i.e., sequences of stops) on a given road network to maximize user experience, typically measured through average travel time, the number of transfers, and network coverage~\cite{mandl1980evaluation}.
More recently, route design is often coupled with frequency setting, since service frequency influences both waiting times and effective line capacities.

The combinatorial nature of transit network design stems from the need to generate candidate routes.
The space of potential routes is enormous, and this has motivated a number of heuristics for path generation and selection.
Early heuristic methods~\cite{baaj1991ai,cipriani2012transit,baaj1995hybrid} use demand profiles, practical design rules, and operator experience to generate an initial set of plausible routes, and then apply additional heuristics to select a subset of lines and frequencies.
These methods are computationally efficient, but typically lack a feedback loop that iteratively improves the candidate route set based on the quality of the final solution.

To address this,~\cite{lee2005transit} proposed an iterative procedure that starts from a shortest-path-based route network and gradually trades off in-vehicle time and waiting time, using demand assignment results to update frequencies.
While this methods ensures monotone improvement, it relies on shortest paths for initialization, and can therefore lead to conservative designs with large optimality gaps, requiring many iterations to reach high-quality solutions.

Another line of work formulates transit network and frequency designs as mathematical programs.
Compared to pure heuristics, these formulations offer formal optimality guarantees or bounds, but they are challenging to solve at city scale~\cite{claessens1998cost}.
As a consequence, many studies decompose decisions into multiple phases, for instance, restricting to a predefined subset of transit routes and the solving for route selection and frequencies~\cite{guan2006simultaneous,goossens2004branch,cancela2015mathematical}.
In such approaches, the quality of the initial route set largely determines the quality of the final design, since routes are not generated adaptively during optimization.

\gls{acr:cg} has been particularly successful in this context.
\cite{borndorfer2007column} is among the first to apply \gls{acr:cg}  to line
planning in public transport. 
\gls{acr:cg} decomposes the problem into a master problem and a pricing problem, where the pricing subproblem generates new routes (columns) that can improve the objective. 
This yields a principled way to combine route generation with route selection and frequency design in a single framework. 
Follow-up works extend this idea to jointly consider line planning and passenger routing~\cite{borndorfer2012direct,jin2016optimizing,bertsimas2021data}, and to incorporate additional constraints and quality criteria. 
These \gls{acr:cg}-based approaches are less sensitive to the initial line set, provide anytime solutions with explicit optimality gaps, and have demonstrated scalability on large transit networks; see the survey~\cite{duran2022survey} for a recent overview.

\subsection{AMoD network design and research gap}
\label{subsec:related-amod}
Conceptually, the \gls{acr:andp} shares high-level similarities
with both road and transit network design: in all three cases a
network designer chooses a service network topology, and travelers (or vehicles) are subsequently routed on that network. 
However, the service paradigm and behavioral assumptions underlying \gls{acr:amod} lead to fundamental differences along three main dimensions.

\paragraph*{Centralized control vs. user equilibrium}
In classical road network design, travelers use their own vehicles and retain full autonomy over route choice.
This motivates the use of \gls{acr:ue}/\gls{acr:sue} models and leads to a bi-level formulation in which the upper-level regulator chooses infrastructure and the lower level captures the equilibrium response of self-interested users.

In contrast, \gls{acr:amod} systems rely on centrally managed fleets of \glspl{acr:av}. 
\glspl{acr:av} are owned and dispatched by the operator, who is responsible both for deciding which links to instrument and for routing all vehicles and passengers on the resulting operation network. 
Passenger routes are not the outcome of decentralized user optimization, but the result of a centralized assignment optimized to maximize the operator’s objective. 
As a consequence, the \gls{acr:andp} does \emph{not} exhibit the bilevel structure characteristic of road network design, but rather a single-level mixed-integer formulation in which infrastructure and routing decisions are jointly optimized.

\paragraph*{Profit-driven, route-level quality constraints}
Transit network design is typically motivated by social welfare and public service objectives. 
Design criteria emphasize system-wide efficiency, affordability, and sustainability, often captured through aggregate metrics such as average travel time, total operating cost, or network coverage~\cite{duran2022survey}.
These aggregate measures can mask poor service for small subsets of travelers; only a limited subset of works explicitly enforce route-level guarantees, and doing so usually requires heavy heuristics to remain tractable~\cite{baaj1995hybrid,bertsimas2021data}.

By contrast, \gls{acr:amod} operators are profit-driven entities competing in a market with low switching costs. 
In such a setting, passenger-level service quality (e.g., door-to door travel time, number of risky maneuvers, or reliability) is
directly tied to demand retention and profitability. 
It is therefore natural to impose route-level \gls{acr:qos} and risk-related constraints, ensuring that \emph{every} passenger route satisfies minimum service criteria. 
The \gls{acr:andp} studied in this paper explicitly incorporates such path-level constraints, which substantially complicate the formulation and preclude purely link-based models.

\paragraph*{Service paradigm and modeling implications}
Finally, the underlying service paradigm differs. 
Transit systems are based on fixed lines, schedules, and transfers: waiting times and transfers are central determinants of user experience and are key modeling elements. 
\gls{acr:amod}, in contrast, offers on-demand, door-to-door service with no transfers, so the primary levers are network instrumentation, vehicle routing, and fleet sizing~\cite{zardini2022co}.
Transfers and waiting time play a negligible role in the \gls{acr:amod} context, but route-level travel time and risk metrics become central.

\paragraph*{Positioning of this work}
The above discussion highlights the distinctive nature of the \gls{acr:andp}: a centralized operator jointly designs the service network and the routing of passengers and vehicles, subject to explicit route-level quality and risk constraints.
Classical road network design assumes self-interested users and bi-level \gls{acr:ue}; 
transit network design is organized around fixed lines, schedules, and transfers, often with social-welfare objectives and aggregate performance
measures. 
Existing \gls{acr:amod} research, in turn, has focused
primarily on operational control (e.g., matching and
rebalancing) on a fixed road network~\cite{zardini2022analysis}.

In contrast, the problem studied in this paper combines elements of all three areas: it is a network design problem, but the decision maker is a profit-driven \gls{acr:amod} operator with centralized control and route-level \gls{acr:qos} and risk constraints.
This leads to a single-level, path-based \gls{acr:milp} that differs structurally from bi-level road network design and from line-based transit formulations. 
The remainder of the paper develops this \gls{acr:andp} formulation, proposes a scalable \gls{acr:cg}-based algorithm with optimality-gap guarantees, and evaluates the resulting framework on city-scale \gls{acr:amod} applications.

\section{System Model and AMoD Network Design Problem}
\label{sec:system-model}
In this section, we formalize the setting for the \gls{acr:andp}.
We first introduce the base road network and demand model, then define paths and path-based specifications.
Building on these preliminary concepts, we formulate a general \gls{acr:andp} and state the structural assumptions on the objective that will be used throughout the paper.
We conclude with the specific profit-maximizing \gls{acr:andp} instance that is the focus of the remainder of the work.

\subsection{Base Road Network and Demand}
\label{subsec:base-network-demand}
The \gls{acr:amod} operator designs a service \emph{operation network} within a given road infrastructure and centrally routes \glspl{acr:av} and passengers on such network.

\begin{definition}[Base Road Network]
The \emph{base road transportation network} is represented by a directed graph~$\graph=\tup{\nodes, \edges}$, where~$\nodes$ is the set of nodes and~$\edges \subseteq \nodes\times \nodes$ is the set of directed edges.
For each edge~$e\in \edges$, let~$s(e)$ and~$t(e)$ denote its source and sink nodes, respectively, and assume that there is at most one edge between any ordered pair of nodes.

Each edge~$e\in \edges$ is associated with a vector of time-invariant attributes
\[
A(e) = \big(A_1(e),\dots,A_n(e)\big) \in
\mathcal A_1 \times \dots \times \mathcal A_n,
\]
where~$A_i : \edges \to \mathcal A_i$ returns attribute~$i$ of edge~$e$. 
In particular, $T : \edges \to \mathbb R_{>0}$ and~$L : \edges \to \mathbb R_{>0}$ denote the travel time and length of edge~$e$.
\end{definition}

\begin{definition}[Demand]
\label{def:demand}
For any ordered pair of distinct nodes~$\tup{i,j} \in \nodes \times \nodes$ with~$i \neq j$, let~$d_{ij} := \tup{i,j,\alpha_{ij}}$ denote the travel demand from origin~$i$ to destination~$j$, where~$\alpha_{ij} \in \mathbb{R}_{>0}$ is the number of travelers wishing to move from~$i$ to~$j$ during the planning horizon.
The total demand is the set
\[
D := \{\,d_{ij} \mid i,j \in \nodes,\ i \neq j,\ \alpha_{ij} > 0\,\}.
\]
We write
\[
\tilde D := \{(i,j) \in \nodes \times \nodes \mid d_{ij} \in D\}
\]
for the corresponding set of OD node pairs.
\end{definition}

\gls{acr:av}s travel along paths in the road network.

\begin{definition}[Path]
\label{def:path}
Let~$i,j \in \nodes$ be distinct nodes. 
A path from~$i$ to~$j$ can be represented in either of the following equivalent forms:
\begin{enumerate}
    \item A node sequence~$p := \tup{v_1,\dots,v_\ell}$ such that~$v_1 = i$,~$v_\ell = j$, and~$\tup{v_k,v_{k+1}} \in E$ for all~$k = 1,\dots,\ell-1$.
    \item An edge sequence~$p := \tup{e_1,\dots,e_{\ell-1}}$ such
    that~$t(e_k) = s(e_{k+1})$ for all~$k = 1,\dots,\ell-2$.
\end{enumerate}
\end{definition}

\begin{remark}
\label{rem:path-representation}
The node-sequence and edge-sequence representations are in one-to-one correspondence: each node path uniquely determines an edge path, and vice versa.
\end{remark}

\begin{definition}[Demand-satisfactory Path]
\label{def:demand-satisfactory-path}
A path~$p$ is \emph{demand-satisfactory} for OD pair~$\tup{i,j} \in \tilde D$ if it starts at~$i$ and ends at~$j$. 
The set of demand-satisfactory paths for~$\tup{i,j}$ is denoted $P_{ij}$, and we write
\[
P := \bigcup_{\tup{i,j}\in \tilde D} P_{ij}
\]
for the union of all demand-satisfactory paths.
\end{definition}

The operation network induced by a design~$\des$ is a subgraph of~$\graph$ that contains only the edges on which the \gls{acr:amod} service is allowed to operate.
This notion will be made precise in \cref{def:andp}.

\subsection{Paths and Path-based Quality-of-Service Constraints}
\label{subsec:path-qos}

Many service specifications for \gls{acr:amod} are naturally expressed at the path level, e.g., bounds on end-to-end travel time or limits on risky maneuvers.
To capture these, we associate attributes via paths.
For a path~$p\in P_{ij}$, its travel time and length are~$T(p):=\sum_{e\in p}T(e)$ and~$L(p):=\sum_{e\in p}L(e)$, respectively.
More generally, let~$H(p)=\tup{H_1(p),\ldots, H_m(p)}$ be a vector of path attributes, where each~$H_k\colon P\to \mathbb{R}$ may depend on the sequence of edges in~$p$ (e.g., number of left turns, number of intersections, or a risk score).

A generic path-level \gls{acr:qos} or risk constraint for OD pair~$\tup{i,j}$ can then be written as
\begin{equation}
    \label{eq:path-constraint-generic}
    H_k(p)\leq M^{k}_{ij}, \qquad \forall p\in P_{ij},
\end{equation}
for prescribed thresholds~$M^{k}_{ij}$.
In the simplest cast,~$H_1(p)=T(p)$ encodes a bound on end-to-end travel time;
in \cref{sec:left-turn-constraint} we will instantiate~$H_k$ to count left turns and obtain a risk-related constraint.

The path admissible set for OD~$\tup{i,j}$ under a collection of constraints of the form~\cref{eq:path-constraint-generic} is
\begin{equation*}
    P_{ij}^{\text{adm}} :=
\big\{\,p \in P_{ij} \,\big|\, H_k(p) \le M^{k}_{ij}\ 
\text{for all relevant } k \big\}.
\end{equation*}
These admissible paths will be the building blocks for the flow variables in the path-based formulation introduced later (\cref{subsec:path-formulation}).

\subsection{General AMoD Network Design Problem}
\label{subsec:general-andp}
We now formalize the \gls{acr:andp} at a high level.
The operator chooses (i) which edges to instrument for autonomous operation, and (ii) a feasible flow assignment on the resulting operation network that serves a subset of the demand.

\begin{definition}[AMoD--NDP]
\label{def:andp}
Given a base network~$\graph=\tup{\nodes,\edges}$ and demand profile~$D$, the \gls{acr:andp} seeks to maximize a system functionality~$F(\des,\flow)$ over:
\begin{itemize}
    \item a binary instrumentation vector~$\des = (x_e)_{e\in E}$, where~$x_e = 1$ if edge~$e$ is instrumented for autonomous operation and~$x_e=0$ otherwise;
    \item a nonnegative passenger flow vector~$\flow = (y_e)_{e\in E}$, where~$y_e$ is the total passenger flow routed along edge~$e$.
\end{itemize}
The decisions~$\tup{\des,\flow}$ are constrained by:
\begin{align*}
    \sum_{e\in \edges} b_e x_e &\le B
    \qquad \text{(instrumentation budget)},\\
    \sum_{e\in \edges} r_e y_e &\le R
    \qquad \text{(resource / fleet-time budget)},\\
    0 \le y_e &\le c_e x_e, \forall e \in \edges
    \quad\text{(edge capacity)}, \\
    \flow \text{ obeys} & \text{ flow conservation} \quad \text{(flow feasibility)}, \\
    \flow \text{ serves} & \text{ at most the demand~$D$} \quad \text{(maximum demand)},
\end{align*}
where $b_e > 0$ is the infrastructure cost of instrumenting edge~$e$,~$c_e > 0$ is its capacity, and~$r_e > 0$ is the resource consumed per unit flow on~$e$ (e.g., vehicle time).

The induced operation network is
\[
\graph^{\mathrm{op}}(\des) = \big(\nodes, \edges^{\mathrm{op}}(\des)\big),
\qquad
\edges^{\mathrm{op}}(\des) := \{e \in \edges \mid x_e = 1\}.
\]
\end{definition}

In general, the functionality~$F(\des,\flow)$ may depend on both decisions. 
In this paper, we focus on a class of models where~$\des$ affects performance only through the feasible set of $\tup{\des,\flow}$ and the objective is additive over edges.

\begin{assumption}[Link-separable functionality]
\label{ass:link-separable}
The functionality depends on the flows via a link-separable
function
\[
F(\des,\flow) = \sum_{e\in E} F_e(y_e),
\]
and~$\des$ influences performance only by restricting the feasible set of $\tup{\des,\flow}$ through the budget, capacity, and flow conservation constraints.
\end{assumption}

\begin{assumption}[Affine per-edge contribution]
\label{ass:affine}
For each edge~$e \in \edges$, the per-edge function~$F_e$ is affine:
\[
F_e(y_e) = \beta_{e,0} + \beta_e\, y_e,
\]
where~$\beta_e \in \mathbb R$ is the marginal contribution of one unit of flow on edge~$e$ and~$\beta_{e,0} \in \mathbb R$ is a constant intercept.
\end{assumption}

\cref{ass:affine} is standard in system design phases where the performance measure is approximated by aggregated usage with a constant marginal benefit (or cost) per unit flow over the operating range of interest~\cite{ahuja1988network,sheffi1985urban}. 
It includes, for example, total distance traveled, time- or distance-proportional operating costs, and profit objectives built from edge-level revenue and cost components.

The intercept terms in \cref{ass:affine} play no role in the optimal solution, which allows us to simplify the notation.

\begin{lemma}[Irrelevance of intercepts]
\label{lem:intercept}
Under \cref{ass:link-separable} and \cref{ass:affine}, the intercepts~$\beta_{e,0}$ do not affect the set of optimal solutions of the \gls{acr:andp}.
Equivalently, any optimizer for the problem with objective~$\sum_{e\in \edges} \beta_e y_e$ is also optimal for the problem with objective~$\sum_{e\in \edges} (\beta_{e,0} + \beta_e y_e)$, and vice versa.
\end{lemma}

\begin{proof}
Let~$\mathcal F$ denote the feasible set of all~$\tup{\des,\flow}$ that satisfy the budget, capacity, and flow-conservation constraints.
Consider the two optimization problems
\[
\max_{\tup{\des,\flow}\in \mathcal F} \ \sum_{e\in \edges} \beta_e y_e
\quad\text{and}\quad
\max_{\tup{\des,\flow}\in \mathcal F} \
\sum_{e\in \edges} (\beta_{e,0} + \beta_e y_e).
\]
For any feasible $\tup{\des,\flow}\in \mathcal F$ we have
\[
\sum_{e\in \edges} (\beta_{e,0} + \beta_e y_e)
= \sum_{e\in \edges} \beta_{e,0} + \sum_{e\in \edges} \beta_e y_e.
\]
The term~$\sum_{e\in \edges} \beta_{e,0}$ is constant over~$\mathcal F$ and therefore does not change the ordering of feasible solutions. 
Hence both problems have the same set of maximizers.
\end{proof}

In the remainder of the paper, we therefore drop the intercept terms and write
\[
F(\des,\flow) = \sum_{e\in E} \beta_e y_e
\]
without loss of generality.

\subsection{Profit-maximizing AMoD-NDP with Fleet and Infrastructure Constraints}
\label{subsec:profit-instance}

We now specialize the general model to the concrete instance studied in the rest of the paper. 
In this instance, the system functionality corresponds to the expected operator profit over the planning horizon.

For each edge~$e \in \edges$, let~$\beta_e$ denote the expected net profit per unit passenger flow on~$e$ (fare minus variable operating costs). 
The operator chooses which edges to instrument and how to route demand so as to maximize total profit subject to:
\begin{itemize}
    \item an infrastructure budget~$B$ for constructing or
    enabling autonomous-driving infrastructure,
    \item a fleet-time budget~$R$ capturing the available
    vehicle-hours, which is directly related to the fleet size,
    \item link capacities~$c_e$.
\end{itemize}

In this setting, and under \cref{lem:intercept}, the \gls{acr:andp} takes the form

\begin{equation}
\label{eq:amod-ndp-profit}
\begin{aligned}
\max_{\des,\flow} \quad &
\sum_{e\in \edges} \beta_e y_e\\
\text{s.t.}\quad &
\sum_{e\in \edges} b_e x_e \le B,\\
& \sum_{e\in \edges} T(e)\, y_e \le R,\\
& 0 \le y_e \le c_e x_e, \qquad \forall e \in \edges,\\
& \flow \text{ induced by a feasible flow assignment on }
\graph^{\mathrm{op}}(x)\\
& \hspace{0.65cm}\text{that serves at most the demand } D,\\
& x_e \in \{0,1\},\quad y_e \ge 0,\qquad \forall e \in \edges.
\end{aligned}
\end{equation}

The fleet-time budget~$R$ is modeled by taking the resource consumption~$r_e$ in \cref{def:andp} equal to the edge travel time~$T(e)$.
In \cref{sec:methodology}, we refine the aggregate edge flows~$\flow$ into OD-specific and
path-specific flows, derive link- and path-based formulations of~\eqref{eq:amod-ndp-profit}, and develop a \gls{acr:cg}-based algorithm that scales to real-world city networks while accommodating path-based \gls{acr:qos} and risk constraints of the form described in \cref{subsec:path-qos}.

\subsection{Discussion}
In this work we focus on a static, steady‑state design problem. 
The passenger flows represent passenger‑carrying vehicle flow aggregated over a planning horizon (one day in the case study), and the fleet‑time budget captures the total vehicle‑hours required to provide this service.
We deliberately abstract away time‑of‑day variations, short‑term queueing, passenger waiting times, and explicit empty‑vehicle rebalancing flows. 
These phenomena are central for operational fleet control, but incorporating them would require a time‑expanded network with many more nodes and state variables. 
Here we view the \gls{acr:andp} as providing a medium‑term infrastructure and fleet‑sizing plan, to be complemented in practice by more detailed operational policies (matching, rebalancing, charging) applied ex post on the designed operation subnetwork~\cite{zardini2022analysis}.

\section{Path-based Formulation and Column-Generation Algorithm}
\label{sec:methodology}

In this section, we instantiate the profit-maximizing \gls{acr:andp} introduced in \cref{sec:system-model} into explicit link- and path-based mixed-integer linear formulations.
We show that a path-based formulation is necessary to capture route-level \gls{acr:qos} constraints, but leads to an exponential number of potential paths.
To address this, we develop a scalable \gls{acr:cg}-based algorithm that dynamically generates only the relevant paths via a pricing problem cast as a \gls{acr:sprc} and solved exactly by a label-correcting algorithm.

Throughout this section we focus on the profit-maximizing instance~\eqref{eq:amod-ndp-profit}, and we assume that for each \gls{acr:od} pair~$\tup{o,d} \in \tilde D$ the admissible paths are those with travel time at most~$M_{od}$, i.e.,~$P_{od}^{\mathrm{adm}}
:= \{\,p \in P_{od} \mid T(p) \le M_{od}\,\}.$

Other path-level constraints can be handled analogously (see \cref{sec:left-turn-constraint} for an example of left-turn risk constraint).

\subsection{Link-based Formulation}
\label{subsec:link-formulation}
We first present a standard link-based formulation of the profit-maximizing \gls{acr:andp}.

Recall the set of \gls{acr:od} node pairs~$\tilde D$ from \cref{def:demand}.
For each node~$v\in \nodes$, denote its sets of outgoing and incoming edges by
\[
\sigma^+(v) := \{e \in \edges \mid s(e) = v\},\quad
\sigma^-(v) := \{e \in \edges \mid t(e) = v\}.
\]

\begin{formulation}[Link-based \gls{acr:andp}]
\label{form:link}
Define the decision variables
\begin{align*}
x_{ij} &=
\begin{cases}
1, & \text{if edge } \tup{i,j} \text{ is instrumented},\\
0, & \text{otherwise},
\end{cases}\\[1ex]
f_{od} &\in [0,\alpha_{od}] \qquad \text{served demand from $o$ to $d$},\\
f^{od}_{ij} &\ge 0 \qquad \text{flow from $o$ to $d$ routed on edge $\tup{i,j}$}.
\end{align*}
The link-based formulation of the profit-maximizing
\gls{acr:andp} is
\begin{equation}
\label{eq:link-formulation}
\begin{aligned}
\max_{\des,\flow} &
\sum_{\tup{o,d}\in \tilde D} \sum_{\tup{i,j}\in \edges} \beta_{ij} f^{od}_{ij} \\[0.3em]
\text{s.t.}\quad
& \sum_{\tup{i,j}\in \sigma^+(v)} f^{od}_{ij}
  - \sum_{\tup{i,j}\in \sigma^-(v)} f^{od}_{ij}
\\[-0.1em]
& \hspace{3.3em}=
\begin{cases}
 f_{od}, & v = o,\\
 -f_{od}, & v = d,\\
 0, & \text{otherwise},
\end{cases}
&& \forall \tup{o,d}\in \tilde D,\ \forall v\in V,\\[0.3em]
& 0 \le f_{od} \le \alpha_{od},
&& \forall \tup{o,d}\in \tilde D,\\[0.3em]
& \sum_{\tup{o,d}\in \tilde D} f^{od}_{ij}
  \le c_{ij} x_{ij},
&& \forall \tup{i,j}\in \edges,\\[0.3em]
& \sum_{\tup{i,j}\in \edges} b_{ij} x_{ij} \le B,\\[0.3em]
& \sum_{\tup{i,j}\in \edges} T\tup{i,j}\,
   \sum_{\tup{o,d}\in \tilde D} f^{od}_{ij}
   \le R,\\[0.3em]
& x_{ij} \in \{0,1\},\quad f^{od}_{ij} \ge 0.
\end{aligned}
\end{equation}
\end{formulation}

The objective of~\eqref{eq:link-formulation} maximizes expected profit.
Constraints enforce flow conservation for each OD pair, demand bounds, link capacities, the infrastructure budget~$B$, and the fleet-time budget~$R$.

\paragraph*{Limitations of the link-based formulation}
\cref{form:link} is compact, but it encodes the service entirely through \emph{edge flows}. This creates two fundamental limitations.
First, many service guarantees for \gls{acr:amod} are naturally path-based, such as bounds on door-to-door travel time or limits on risky maneuvers (e.g., left turns). 
These cannot in general be expressed solely as functions of edge flows~$\{f^{od}_{ij}\}$, because they depend on the \emph{sequence} of edges traversed.
Second, even when the objective is link-separable, as in \cref{ass:link-separable}, link-based decision variables make it difficult to integrate additional path-dependent metrics or constraints without substantially complicating the model.

These limitations motivate a path-based formulation in which flows are defined on \gls{acr:od}-specific paths rather than on edges.

\subsection{Path-based Formulation and Equivalence}
\label{subsec:path-formulation}

For each~$\tup{o,d}\in \tilde D$, recall the set~$P_{od}$ of demand-satisfactory paths from \cref{def:demand-satisfactory-path}, and define the admissible set
\[
P_{od}^{\mathrm{adm}} :=
\{\,p \in P_{od} \mid T(p) \le M_{od}\,\}.
\]
Let~$P^{\mathrm{adm}} := \bigcup_{\tup{o,d}\in \tilde D} P_{od}^{\mathrm{adm}}$ denote the union over all \gls{acr:od} pairs.

For each path~$p\in P_{od}^{\mathrm{adm}}$ we introduce:
\begin{itemize}
    \item a path-flow variable~$f^{p}_{od} \ge 0$ representing the amount of demand~$\tup{o,d}$ routed along path~$p$;
    \item a binary incidence parameter~$z^{od,p}_{ij}$ equal to $1$ if edge~$\tup{i,j}$ belongs to~$p$ and~$0$ otherwise;
    \item path profit and travel time
    \[
    \beta^p_{od} := \sum_{e\in p} \beta_e, \qquad
    t_p := \sum_{e\in p} T(e).
    \]
\end{itemize}

\begin{formulation}[Path-based \gls{acr:andp}]
\label{form:path}
The path-based formulation of the profit-maximizing
\gls{acr:andp} is
\begin{equation}
\label{eq:path-formulation}
\begin{aligned}
\max_{\des,\flow} \quad &
\sum_{\tup{o,d}\in \tilde D} \sum_{p\in P_{od}^{\mathrm{adm}}}
\beta^p_{od} f^p_{od} \\[0.3em]
\text{s.t.}\quad
& \sum_{p\in P_{od}^{\mathrm{adm}}} f^p_{od}
  \le \alpha_{od},
&& \forall \tup{o,d}\in \tilde D,\\[0.3em]
& \sum_{\tup{o,d}\in \tilde D} \sum_{p\in P_{od}^{\mathrm{adm}}}
  z^{od,p}_{ij} f^p_{od}
  \le c_{ij} x_{ij},
&& \forall \tup{i,j}\in \edges,\\[0.3em]
& \sum_{\tup{i,j}\in \edges} b_{ij} x_{ij} \le B,\\[0.3em]
& \sum_{\tup{o,d}\in \tilde D} \sum_{p\in P_{od}^{\mathrm{adm}}}
  t_p f^p_{od} \le R,\\[0.3em]
& x_{ij} \in \{0,1\},\quad f^p_{od} \ge 0.
\end{aligned}
\end{equation}
\end{formulation}

The path-based formulation naturally accommodates path-level \gls{acr:qos} constraints.
For instance, the travel-time constraints~$T(p) \le M_{od}$ are enforced by restricting the admissible path sets~$P_{od}^{\mathrm{adm}}$. 
In \cref{sec:left-turn-constraint} we will augment \cref{form:path} with an additional risk-related constraint on the total number of left turns.

When no explicit path-dependent constraints are present, the link- and path-based formulations are equivalent.

\begin{lemma}[Equivalence without path constraints]
\label{lem:link-path-equivalence}
Suppose that~$P_{od}^{\mathrm{adm}} = P_{od}$ for all~$\tup{o,d}\in\tilde D$, i.e., no path-level constraints are imposed.
Then \cref{form:link} and~\cref{form:path} are equivalent: there exists a bijection between their feasible solutions that preserves the objective value.
\end{lemma}

\begin{proof}
Given a feasible solution~$\tup{\des,\flow}$ to the path-based formulation, define for each \gls{acr:od} pair and edge
\[
\tilde f^{od}_{ij}
:= \sum_{p\in P_{od}} z^{od,p}_{ij} f^p_{od}.
\]
These flows satisfy the flow-conservation constraints of \cref{form:link}, because for each~$\tup{o,d}$ the path flows form a feasible decomposition of~$f_{od} := \sum_{p\in P_{od}} f^p_{od}$ into~$o$-$d$ paths. 
Edge capacities, budgets, and fleet-time constraints coincide under the two formulations, since
\[
\sum_{\tup{o,d}\in \tilde D} \tilde f^{od}_{ij}
= \sum_{\tup{o,d}\in \tilde D} \sum_{p\in P_{od}}
  z^{od,p}_{ij} f^p_{od},
\]
and similarly for total travel time. 
The objective values match:
\[
\sum_{\tup{o,d}\in \tilde D} \sum_{\tup{i,j}\in \edges}
\beta_{ij} \tilde f^{od}_{ij}
=
\sum_{\tup{o,d}\in \tilde D} \sum_{p\in P_{od}}
\beta^p_{od} f^p_{od}.
\]

Conversely, given a feasible link-based solution~$\tup{\des,\tilde \flow}$, fix~$\tup{o,d}\in \tilde D$ and decompose the~$o$-$d$ flow induced by~$\{\tilde f^{od}_{ij}\}_{\tup{i,j}\in \edges}$ into a nonnegative combination of~$o$-$d$ paths via standard flow decomposition.
Assigning the corresponding path flows~$f^p_{od}$ recovers a feasible solution to
\cref{form:path} with the same edge loads and total profit. 
This defines a bijection between feasible solutions of the two formulations that preserves the objective value.
\end{proof}

In the remainder of the paper, we work with the path-based \cref{form:path}, which allows us to explicitly encode path-level constraints. However, \cref{form:path} is a \gls{acr:milp} with a potentially exponential number of path variables. 
Directly enumerating all paths~$P^{\mathrm{adm}}$ is infeasible on city-scale networks. To address this issue, we develop a \gls{acr:cg}-based algorithm that generates paths only when they can improve the objective. An overview of the complete algorithm is presented in \cref{fig:algo}. In the remainder of this section, we explain each step in detail.

\begin{figure*}[tb]
    \centering
    \includegraphics[width=0.95\linewidth]{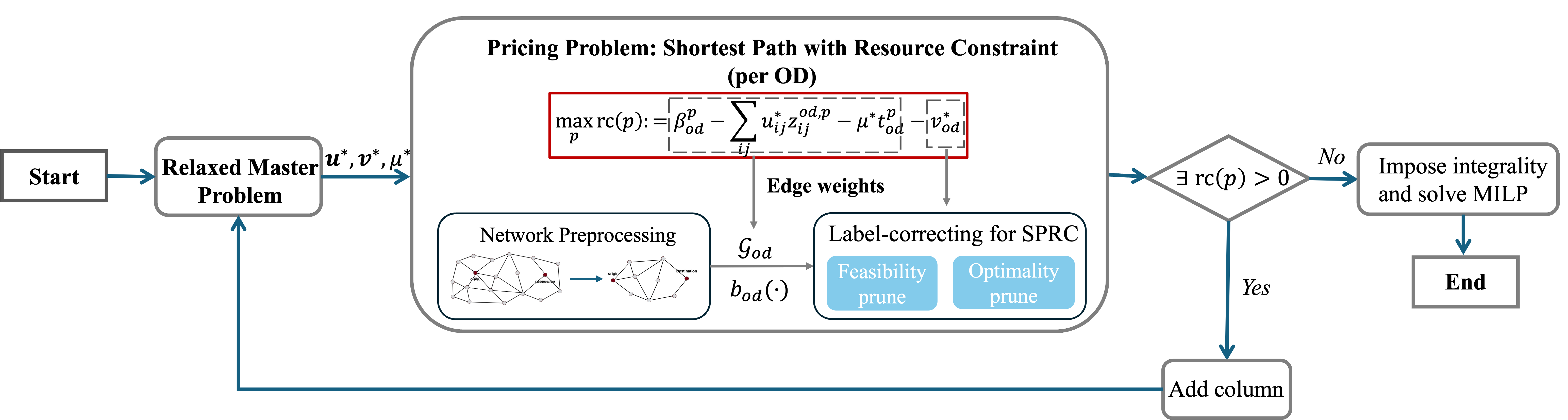}
    \caption{\gls{acr:cg}-based algorithm to solve the \gls{acr:andp} based on path-based formulation. The algorithm relies on a \gls{acr:cg}-based decomposition that decomposes \cref{form:path} into a master and pricing problem. The master and pricing problem is solved iteratively until the convergence condition is satisfied.}
    \label{fig:algo}
\end{figure*}

\subsection{LP Relaxation and Column-Generation Framework}
\label{subsec:cg-framework}

We start by relaxing~$x_{ij}\in\{0,1\}$ to $x_{ij}\in[0,1]$. 
The resulting \emph{master problem} in standard form is:

\begin{problem}[LP master problem]
\label{prob:mp}
\begin{equation}
\label{eq:mp}
\begin{aligned}
\max_{\des,\flow} \quad &
\sum_{\tup{o,d}\in \tilde D} \sum_{p\in P_{od}^{\mathrm{adm}}}
\beta^p_{od} f^p_{od} \\[0.3em]
\text{s.t.}\quad
& \sum_{p\in P_{od}^{\mathrm{adm}}} f^p_{od}
  \le \alpha_{od},
&& \forall \tup{o,d}\in \tilde D,\\[0.3em]
& \sum_{\tup{o,d}\in \tilde D} \sum_{p\in P_{od}^{\mathrm{adm}}}
  z^{od,p}_{ij} f^p_{od}
  - c_{ij} x_{ij} \le 0,
&& \forall \tup{i,j}\in \edges,\\[0.3em]
& \sum_{\tup{i,j}\in \edges} b_{ij} x_{ij} \le B,\\[0.3em]
& \sum_{\tup{o,d}\in \tilde D} \sum_{p\in P_{od}^{\mathrm{adm}}}
  t_p f^p_{od} \le R,\\[0.3em]
& 0 \le x_{ij} \le 1,\quad f^p_{od} \ge 0.
\end{aligned}
\end{equation}
\end{problem}

Let~$v_{od}$,~$u_{ij}$,~$\pi$,~$\mu$, and~$\delta_{ij}$ denote dual variables associated with the demand, capacity, budget, fleet-time, and upper-bound constraints, respectively. The dual of~\eqref{eq:mp} is:

\begin{problem}[Dual of LP master problem]
\label{prob:dual}
\begin{equation}
\label{eq:dual}
\begin{aligned}
\min_{\mu,v,\pi,u,\delta} \quad &
\sum_{\tup{o,d}\in\tilde D} v_{od}\alpha_{od}
+ B\pi
+ R\mu
+ \sum_{\tup{i,j}\in \edges} \delta_{ij} \\[0.3em]
\text{s.t.}\quad
& \sum_{\tup{i,j}\in \edges} u_{ij} z^{od,p}_{ij}
  + v_{od} + \mu t_p
  \ge \beta^p_{od},
\hspace{0.4em} \begin{aligned}[t]
   &\forall \tup{o,d}\in \tilde D,\\
   &\forall p\in P_{od}^{\mathrm{adm}},
\end{aligned}\\[0.3em]
& -c_{ij} u_{ij} + b_{ij} \pi + \delta_{ij} \ge 0,
\hspace{0.6em} \forall \tup{i,j}\in \edges,\\[0.3em]
& u_{ij},v_{od},\mu,\pi,\delta_{ij} \ge 0.
\end{aligned}
\end{equation}
\end{problem}

In a \gls{acr:cg} scheme, we do not include all paths~$P_{od}^{\mathrm{adm}}$ in the master problem. 
Instead, we maintain a restricted subset~$P_{od}^r \subseteq P_{od}^{\mathrm{adm}}$ and solve the \emph{restricted master problem} (\gls{acr:rmp}) over~$P^r := \bigcup_{(o,d)} P_{od}^r$, obtaining primal and dual optimal solutions. 
We then ask whether there exists any omitted path~$p \in P_{od}^{\mathrm{adm}}\setminus P_{od}^r$ with strictly positive reduced cost; such a path would improve the master problem objective and should be added to $P^r$ as a new column.

For a path-flow variable~$f^p_{od}$ the reduced cost is
\begin{equation}
\label{eq:reduced-cost}
\operatorname{rc}(p)
= \beta^p_{od}
  - \sum_{\tup{i,j}\in \edges} u_{ij}^\star z^{od,p}_{ij}
  - v_{od}^\star
  - \mu^\star t_p,
\end{equation}
where~$\tup{u^\star,v^\star,\mu^\star,\pi^\star,\delta^\star}$ is an optimal solution to the dual problem (\cref{prob:dual}) for the current \gls{acr:rmp}. 
A path with~$\operatorname{rc}(p) > 0$ can improve the objective and is therefore a promising candidate.
Identifying such paths leads to the \emph{pricing problem}, which we discuss next.

\subsection{Pricing as Shortest Path with Resource Constraint}
\label{subsec:sprc}
For a fixed \gls{acr:od} pair $\tup{o,d}$, the pricing problem consists of finding a path~$p\in P_{od}^{\mathrm{adm}}$ that maximizes the reduced cost~\eqref{eq:reduced-cost}. 
Because~$v_{od}^\star$ is constant for a given \gls{acr:od}, it can be dropped from the optimization without changing the optimizer. 
Using~$\beta^p_{od} = \sum_{e\in p} \beta_e$ and~$t_p = \sum_{e\in p} T(e)$, we can rewrite \eqref{eq:reduced-cost} as
\[
\operatorname{rc}(p)
= \sum_{e\in p} \bigl(\beta_e - u_e^\star - \mu^\star T(e)\bigr)
- v_{od}^\star,
\]
where we write~$u_e^\star$ for~$u_{ij}^\star$ when $e = \tup{i,j}$.

Dropping the constant~$-v_{od}^\star$ and changing sign, the pricing problem for~$\tup{o,d}$ is equivalent to the following shortest path problem with resource constraint:
\begin{equation}
\label{eq:sprc}
\begin{aligned}
\min_{p} \quad &
\sum_{e\in p} c_e \\[0.3em]
\text{s.t.}\quad
& p \in P_{od}^{\mathrm{adm}},\\[0.2em]
& \text{$p$ is elementary},
\end{aligned}
\end{equation}
where the edge cost
\[
c_e := \mu^\star T(e) + u_e^\star - \beta_e
\]
encodes the dual penalties associated with travel time and capacity.

\begin{lemma}[Equivalence to SPRC]
\label{prop:sprc}
For each OD pair~$\tup{o,d}$, the pricing problem that maximizes~$\operatorname{rc}(p)$ over~$P_{od}^{\mathrm{adm}}$ is equivalent to the elementary \gls{acr:sprc}~\eqref{eq:sprc} on~$\graph$, where edge~$e$ has cost~$c_e$ and consumes resource~$T(e)$, and the total available resource is ~$M_{od}$. 
If~$p^\star$ is an optimal solution to~\eqref{eq:sprc}, then~$p^\star$ is also an optimal pricing path, with
\[
\operatorname{rc}(p^\star)
= - \sum_{e\in p^\star} c_e - v_{od}^\star.
\]
\end{lemma}

\begin{proof}
The result follows directly from the decomposition of~$\beta^p_{od}$ and~$t_p$ into edge contributions and the expression~\eqref{eq:reduced-cost}. 
Dropping the constant~$v_{od}^\star$ does not change the maximizer; changing the sign transforms maximization into minimization with edge costs~$c_e$.
\end{proof}

Because some~$c_e$ can be negative, negative cycles may exist. 
To prevent ill-posedness due to infinitely improving cycles, we restrict attention to \emph{elementary} paths, which is also consistent with realistic routing in road networks, as cycles incur detours that travelers would not consider~\cite{prato2009route}.

\begin{definition}[Elementary path]
\label{def:elementary-path}
A path with node-sequence representation~$p := \tup{v_1,\dots,v_\ell}$ is \emph{elementary} if all nodes along the path are distinct.
\end{definition}

We therefore need an exact algorithm for the elementary \gls{acr:sprc}. 
We follow the label-correcting paradigm introduced in~\cite{feillet2004exact} and adapt it to leverage the structure of road networks via \gls{acr:od}-specific preprocessing and A$^\star$-style bounds.

\subsection{Pricing Algorithm: Preprocessing and Label-Correcting Search}
\label{subsec:pricing-alg}

The \gls{acr:sprc} for each OD pair~$\tup{o,d}$ is solved in
two steps:
\begin{enumerate}
    \item \textbf{\gls{acr:od}-specific network preprocessing}, which
    prunes nodes and edges that cannot lie on any feasible
    path with~$T(p)\le M_{od}$.
    \item \textbf{Label-correcting search} on the pruned
    network, using the preprocessing results as admissible
    bounds in an A$^\star$-like selection rule.
\end{enumerate}

\paragraph*{1) \gls{acr:od}-specific network preprocessing}
For a given~$\tup{o,d}\in \tilde D$, we run Dijkstra's algorithm forward from~$o$ and backward from~$d$ (on the reversed network) using travel time~$T(e)$ as edge weights. 
Let~$t_{ov}$ denote the shortest-path travel time from~$o$ to node~$v$, and~$t_{vd}$ the shortest-path time from~$v$ to~$d$. 
We then prune any node~$v$ for which~$t_{ov} + t_{vd} > M_{od}$ and remove all incident edges.

\begin{algorithm}
\small
\caption{Network pre-processing}
\label{algo:preprocess}
\begin{algorithmic}[1]
\Require Base network $\mathcal{G}=(\mathcal{V},\mathcal{E})$ with edge travel times $T(e)$; OD pair $(o,d) \in \tilde{D}$; travel time thresholds $M_{od}$
\Ensure OD-specific processed graphs $\mathcal{G}_{od}=(\mathcal{V}_{od},\mathcal{E}_{od})$ and bounds $\{b_{od}(v)=t_{vd}\}_{v\in\mathcal{V}_{od}}$

\State $\mathcal{G}^{\mathsf{rev}} \gets \textsc{ReverseGraph}(\mathcal{G})$
    \State \textcolor{blue}{// Forward search from $o$}
    \State Compute $t_{o\cdot} \gets \textsc{Dijkstra}(\mathcal{G},\, o)$
    \State $\mathcal{R}_F^{od} \gets \{v\in\mathcal{V}: t_{ov}\le M_{od}\}$ \Comment{Prune after forward search}

    \State \textcolor{blue}{// Backward search to $d$ via reversed graph}
    \State $t_{\cdot d} \gets \textsc{Dijkstra}(\mathcal{G}^{\mathsf{rev}},\, d)$
    \State Set $t_{vd}\gets \tilde t_{dv}$ for all $v\in\mathcal{R}_F^{od}$

    \State \textcolor{blue}{// Prune nodes and edges inconsistent with the time budget}
    \State $\mathcal{V}_{od} \gets \{v\in\mathcal{R}_F^{od}: t_{ov}+t_{vd}\le M_{od}\}$
    \State $\mathcal{E}_{od} \gets \{(u,v)\in\mathcal{E}: u,v\in\mathcal{V}_{od}\}$
    \State $b_{od}(v)\gets t_{vd}$ for all $v\in\mathcal{V}_{od}$

\State \textbf{return } $\mathcal{G}_{od}$, $b_{od}(\cdot)$
\end{algorithmic}
\end{algorithm}

\begin{lemma}[Preservation of feasible paths]
\label{lem:preprocess-preserve}
The preprocessing step of \cref{algo:preprocess} preserves all elementary paths~$p$ from~$o$ to~$d$ satisfying~$T(p)\le M_{od}$.
\end{lemma}

\begin{proof}
Let~$p$ be any feasible elementary path from~$o$ to~$d$. By feasibility, we know that the total travel time of the path~$T(p) \le M_{od}$. Let~$T(o\leadsto v)$ and~$T(v\leadsto d)$ denote the travel time of the corresponding subpaths of~$p$. 
Summing yields
\[
T(o\leadsto v) + T(v\leadsto d)
= T(p) \le M_{od}.
\]
Hence~$v$ is not pruned. 
Since this holds for every node on~$p$, the entire path is preserved.
\end{proof}

\paragraph*{2) Label-correcting algorithm for SPRC}
On the processed graph~$\graph_{od}$ we perform a
label-correcting search for the elementary \gls{acr:sprc}. 
A label represents a partial path from~$o$ to some node~$v$.

\begin{definition}[Label]
\label{def:label}
Fix~$\tup{o,d}$ and processed graph~$\graph_{od}=\tup{\nodes_{od},\edges_{od}}$.
A \emph{label} is a quadruple~$\ell = \tup{v, c, t, \nodes_{\mathrm{vis}}}$ where~$v\in \nodes_{od}$ is the current node,~$c\in\mathbb{R}$ is the accumulated edge cost from~$o$ to~$v$,~$t\in\mathbb{R}_{\ge 0}$ is the accumulated travel time from~$o$ to~$v$, and~$\nodes_{\mathrm{vis}}\subseteq \nodes_{od}$ is the set of nodes visited so far along the partial path, used to enforce elementarity.
\end{definition}

Only labels at the same node are comparable. They are compared using a dominance relation.

\begin{definition}[Dominated label]
\label{def:dominated-label}
Fix~$\tup{o,d}$ and a node~$v\in \nodes_{od}$. 
Let~$\ell_1=\tup{v,c_1,t_1,\nodes^1_{\mathrm{vis}}}$ and~$\ell_2=\tup{v,c_2,t_2,\nodes^2_{\mathrm{vis}}}$ be two labels at~$v$. 
We say that~$\ell_1$ \emph{dominates}~$\ell_2$ if
\[
c_1 \le c_2,\quad
t_1 \le t_2,\quad
\nodes^1_{\mathrm{vis}} \subseteq \nodes^2_{\mathrm{vis}},
\]
with at least one inequality strict. 
In this case~$\ell_2$ is called dominated by $\ell_1$.
\end{definition}

Given a label~$\ell=\tup{v,c,t,\nodes_{\mathrm{vis}}}$, we extend it along an outgoing edge~$e=\tup{v,v'}\in \edges_{od}$ to obtain
\[
\ell' = \bigl(v',\, c + c_e,\,
      t + T(e),\, \nodes_{\mathrm{vis}}\cup\{v'\}\bigr).
\]
The extended label is discarded if any of the following holds:
\begin{itemize}
    \item \textbf{Non-elementary}: $v'\in \nodes_{\mathrm{vis}}$.
    \item \textbf{Resource infeasible}:
    $t + T(e) + t_{v'd} > M_{od}$ (no completion to a
    feasible $o$-$d$ path exists).
    \item \textbf{Dominated}: $\ell'$ is dominated by an
    existing label at node $v'$.
\end{itemize}

We maintain for each node~$v$ a set~$B(v)$ of non-dominated labels and a priority queue~$Q$ keyed by~$t + b_{od}(v)$ (an A$^\star$-like evaluation). 
The algorithm terminates when no labels remain to be processed. 
The cheapest label at~$d$ then corresponds to an optimal solution of the \gls{acr:sprc}.

\begin{algorithm}
\small
\caption{Pricing via RCSP: label-correcting with preprocessing bounds}
\label{algo:rcsp}
\begin{algorithmic}[1]
\Require Base network $\mathcal{G}$; \gls{acr:od} pair $(o,d) \in \tilde{D}$; travel time thresholds $\{M_{od}\}_{(o,d)\in\tilde{\mathcal{D}}}$;
OD-specific processed graph $\mathcal{G}_{od}=(\mathcal{V}_{od},\mathcal{E}_{od})$ and bounds $b_{od}(\cdot)$ from \Cref{algo:preprocess},
dual-dependent edge weights $w_e=\mu T(e)+u_e-\beta_e$, dual variable $v^*_{od}$
\Ensure Optimal path $p^*$ and corresponding reduced cost $\text{rc}(p^*)$

    \State $B(v)\gets\emptyset \quad \forall v\in\mathcal{V}_{od}$
    \State $\mathcal{Q}\gets\emptyset$

    \State $\ell_0 \gets (o,\;0,\;0,\;\{o\})$ \Comment{label at original node}
    \State $B(o)\gets\{\ell_0\}$
    \State Push $\ell_0$ into $\mathcal{Q}$ with key $k(\ell_0)\gets 0+t_{od}$

    \While{$\mathcal{Q}\neq\emptyset$}
        \State $\ell \gets \mathsf{pop\_min}(\mathcal{Q})$

        \For{each edge $e=(\ell.v,v')\in\mathcal{E}_{od}$}
            \State \textcolor{blue}{//Prune if not elementary path}
            \If{$v'\in \ell.\mathcal{V}_{\mathsf{vis}}$} \State \textbf{continue} \EndIf
            \State \textcolor{blue}{//Prune if exceeds travel time threshold}
            \If{$\ell.t + T(e) + t_{v'd} > M_{od}$} \State \textbf{continue} \EndIf

            \State $\ell' \gets \Bigl(v',\;\ell.c+w_e,\;\ell.t+T(e),\;\ell.\mathcal{V}_{\mathsf{vis}}\cup\{v'\}\Bigr)$

            \State \textcolor{blue}{//Prune if dominated}
            \If{$\exists \tilde{\ell}\in B(v')$ that dominates $\ell'$} 
                \State \textbf{continue}
            \EndIf
            \State Remove from $B (v')$ all labels dominated by $\ell'$
            \State Insert $\ell'$ into $B(v')$

            \If{$v'\neq d$}
                \State Push $\ell'$ into $\mathcal{Q}$ with key $k(\ell')\gets \ell'.t+t_{v'd}$
            \EndIf
        \EndFor
    \EndWhile

    \State \textcolor{blue}{//Get label with lowest cost}
    \State $\ell^\star \gets \arg\min_{\ell \in B(d)} \; \ell.c$
    \State Reconstruct the corresponding path $p^\star$ from $\ell^\star$
    \State $\text{rc}(p^*)=-\ell^*.c-v_{od}$


\State \textbf{return } $p^*,\,\text{rc}(p^*)$
\end{algorithmic}
\end{algorithm}

\begin{lemma}[Validity of dominance pruning]
\label{lem:dominance}
Let~$\ell_1$ and~$\ell_2$ be two labels at the same node~$v$ for a fixed OD pair~$\tup{o,d}$, with~$\ell_1$ dominating~$\ell_2$ in the sense of \cref{def:dominated-label}. 
Then discarding~$\ell_2$ cannot eliminate any path that is optimal for the \gls{acr:sprc}.
\end{lemma}

\begin{proof}
Consider any feasible completion of~$\ell_2$ to a path from~$v$ to~$d$ that respects the elementarity and time constraints.
Applying the same sequence of extensions to~$\ell_1$ yields another feasible completion, since~$\ell_1$ uses no more time and visits a subset of the nodes of~$\ell_2$. 
Because~$c_1 \le c_2$, the resulting full path from~$\ell_1$ has cost no greater than the path from~$\ell_2$. 
Thus~$\ell_2$ cannot lead to a strictly better feasible solution than~$\ell_1$, and can be safely discarded.
\end{proof}

\begin{lemma}[Optimality of \cref{algo:rcsp}]
\label{lem:label-optimality}
If there exists an optimal elementary path solving the \gls{acr:sprc}~\eqref{eq:sprc}, \cref{algo:rcsp} returns one such path.
\end{lemma}

\begin{proof}
By \cref{lem:preprocess-preserve}, preprocessing preserves all feasible paths. 
During the label-correcting procedure, labels are discarded only if they violate elementarity, violate the resource bound, or are dominated by another label.
The first two cases correspond to partial paths that cannot be extended into feasible solutions. 
By \cref{lem:dominance}, dominated labels cannot yield an optimal solution. 
Hence at least one label corresponding to an optimal path remains in the label sets, and the algorithm selects it at termination.
\end{proof}

\subsection{Robust extension under uncertainty}
\label{subsec:robust}
The formulation and algorithm presented so far assume a deterministic setting. 
In practice, uncertainty is ubiquitous in transportation systems, especially in travel times and demand~\cite{daganzo1977stochastic}.
In this subsection, we show that the proposed framework can incorporate standard robust counterparts for these uncertainties without changing the overall decomposition logic, and with minimal modifications to the inputs of the pricing problem.

We adopt the framework of \gls{acr:ro}, which is one of the most widely used approaches for handling parameter uncertainty in optimization; see, e.g.,~\cite{bertsimas2011theory,ben2002robust}.
\gls{acr:ro} seeks decisions that remain feasible for all realizations of uncertain parameters within a prescribed uncertainty set.

We focus on \emph{box uncertainty}, which is among the most common uncertainty sets in \gls{acr:ro}. 
We consider uncertainty in travel times and in \gls{acr:od} demand, and show that the robust \gls{acr:andp} preserves the structure exploited by the \gls{acr:cg}-based algorithm.

\begin{definition}[Box uncertainty]
Let~$\tilde{a}$ be an uncertain scalar parameter with nominal value~$\bar{a}$. 
Under box uncertainty,~$\tilde{a}$ belongs to
the set
\[
\mathcal{U}
:= \{\tilde{a} : |\tilde{a} - \bar{a}| \le \rho\},
\]
where~$\rho \ge 0$ is the maximum deviation from the nominal value. For a vector of uncertain parameters~$\tilde{\boldsymbol a} \in \mathbb{R}^n$ with nominal value~$\bar{\boldsymbol a} \in \mathbb{R}^n$ and deviation~$\boldsymbol\Delta \in \mathbb{R}^n$, we write~$\tilde{\boldsymbol a} = \bar{\boldsymbol a} + \boldsymbol\Delta$
with elementwise bounds~$|\Delta_i| \le \rho_i$ for~$i=1,\dots,n$. When~$\rho_i=0$ for~$i=1,\dots,n$, the uncertainty counterpart recovers the deterministic formulation.
\end{definition}

\subsubsection
{Uncertainty in travel times}
We first consider uncertainty in link travel times, which induces uncertainty in path travel times.

\begin{definition}[Uncertainty in link travel time]
\label{def:unc-link-time}
Let~$\bar{T}(e)$ denote the nominal travel time on edge~$e \in \edges$. Collect all nomial edge travel times in a vector~$\bar{\boldsymbol T} \in \mathbb{R}^{|\edges|}_{\ge 0}$.
Under box uncertainty, the real edge travel times are represented as
\[
\tilde{\boldsymbol T}
= \bar{\boldsymbol T} + \boldsymbol\Delta^T,
\quad
|\Delta^T_e| \le \rho^T_e, \ \forall e\in \edges,
\]
where~$\boldsymbol\Delta^T \in \mathbb{R}^{|\edges|}$ is the uncertainty vector and~$\boldsymbol\rho^T \in \mathbb{R}^{|\edges|}_{\ge 0}$ bounds its entries.
\end{definition}

Path travel times then inherit uncertainty through the incidence structure of the network.

\begin{definition}[Induced uncertainty in path travel time]
\label{def:unc-path-time}
Let~$P^{\mathrm{adm}}$ be the set of admissible paths and~$\bar{\boldsymbol t}_p \in \mathbb{R}^{|P^{\mathrm{adm}}|}_{\ge 0}$ the associated vector of nominal path travel times, with entries~$\bar{t}_p = \sum_{e\in p} \bar{T}(e)$. 
Let~$A \in \{0,1\}^{|P^{\mathrm{adm}}|\times |\edges|}$ be the path–edge incidence matrix, where~$A_{pe}=1$ if edge~$e$ lies on path~$p$ and~$A_{pe}=0$ otherwise. 
Under \cref{def:unc-link-time},
\[
\tilde{\boldsymbol t}_p
= A \tilde{\boldsymbol T}
= A\bar{\boldsymbol T} + A\boldsymbol\Delta^T
= \bar{\boldsymbol t}_p + \boldsymbol\Delta^p,
\]
where~$\boldsymbol\Delta^p := A\boldsymbol\Delta^T$ is the induced path time uncertainty and each component~$\Delta^p_p$ satisfies~$|\Delta^p_p| \le \rho^p_p := \sum_{e\in p} \rho^T_e$.
\end{definition}

Thus box uncertainty on link travel times induces box uncertainty on path travel times, with path-wise radii~$\rho^p_p$ obtained by summing link radii along the path.

In the path-based formulation (\cref{form:path}), the fleet-time budget appears as
\begin{equation}
\label{eq:time-constraint-det}
\sum_{(o,d)\in \tilde D}
\sum_{p\in P_{od}^{\mathrm{adm}}}
t_p f^p_{od} \le R.
\end{equation}
Collect all path flows into a vector~$\boldsymbol f \in \mathbb{R}^{|P^{\mathrm{adm}}|}_{\ge 0}$ and all path travel times into~$\boldsymbol t_p$. 
Then \eqref{eq:time-constraint-det} can be written compactly as
\begin{equation}
\label{eq:time-constraint-vector}
\boldsymbol t_p^\top \boldsymbol f \le R.
\end{equation}

Under box uncertainty, the true path travel-time vector is~$\tilde{\boldsymbol t}_p = \bar{\boldsymbol t}_p + \boldsymbol\Delta^p$ with~$|\Delta^p_i| \le \rho^p_i$. Robust feasibility requires that \eqref{eq:time-constraint-vector} hold for all realizations~$\tilde{\boldsymbol t}_p$ in the box:
\[
(\bar{\boldsymbol t}_p + \boldsymbol\Delta^p)^\top \boldsymbol f
\le R
\quad \forall \boldsymbol\Delta^p
\text{ with } |\Delta^p_i|\le \rho^p_i.
\]
This is equivalent to
\begin{align}
\max_{|\boldsymbol\Delta^p|\le \boldsymbol\rho^p}
(\bar{\boldsymbol t}_p + \boldsymbol\Delta^p)^\top \boldsymbol f
&\le R \nonumber\\
\Leftrightarrow\quad
\bar{\boldsymbol t}_p^\top \boldsymbol f
+ \max_{|\boldsymbol\Delta^p|\le \boldsymbol\rho^p}
(\boldsymbol\Delta^p)^\top \boldsymbol f
&\le R. \label{eq:ro-time-1}
\end{align}
Since~$\boldsymbol f \ge 0$, the inner maximization separates componentwise:
\[
\max_{|\Delta^p_i|\le \rho^p_i} \Delta^p_i f_i
= \rho^p_i f_i,
\]
and hence
\[
\max_{|\boldsymbol\Delta^p|\le \boldsymbol\rho^p}
(\boldsymbol\Delta^p)^\top \boldsymbol f
= \sum_i \rho^p_i f_i
= (\boldsymbol\rho^p)^\top \boldsymbol f.
\]
Substituting into \eqref{eq:ro-time-1} yields the robust fleet-time constraint
\begin{equation}
\label{eq:time-constraint-robust}
(\bar{\boldsymbol t}_p + \boldsymbol\rho^p)^\top \boldsymbol f
\le R.
\end{equation}
Thus the box-robust counterpart simply replaces each path travel time~$t_p$ by its inflated value
\[
\tilde{t}^{\mathrm{rob}}_p
:= \bar{t}_p + \rho^p_p.
\]

When the uncertainty is induced from link travel times as in \cref{def:unc-link-time}, we can equivalently work
at the edge level: the robust constraint can be rewritten as
\[
\sum_{e\in \edges} \bigl(\bar{T}(e)+\rho^T_e\bigr) y_e \le R,
\]
where~$y_e$ is the aggregate flow on edge~$e$. 
This shows that the robust counterpart is identical to the deterministic constraint with \emph{inflated edge travel times}
\[
\tilde{T}^{\mathrm{rob}}(e) := \bar{T}(e) + \rho^T_e.
\]

\begin{lemma}[Compatibility under travel-time uncertainty]
\label{prop:ro-time}
Consider the robust \gls{acr:andp} with box uncertainty on
link travel times as in \cref{def:unc-link-time}.
The robust counterpart of the fleet-time constraint is obtained by replacing each edge travel time~$T(e)$ by~$\tilde{T}^{\mathrm{rob}}(e) = \bar{T}(e)+\rho^T_e$ in calculating the path travel time in \cref{form:path} and in the master problem (\cref{prob:mp}). 
The LP master, its dual, and the pricing problem remain of the same form, and the pricing problem is still an elementary \gls{acr:sprc} with edge resource~$\tilde{T}^{\mathrm{rob}}(e)$.
\end{lemma}

\begin{proof}
The robust counterpart of the fleet-time constraint is
\eqref{eq:time-constraint-robust}. 
When uncertainty is induced from link times, this constraint is equivalent to replacing~$T(e)$ by $\bar{T}(e)+\rho^T_e$ in the deterministic formulation, as argued above. 
The dual of the robust master problem is therefore identical in structure to \cref{prob:dual}, with~$t_p$ replaced by~$\tilde{t}^{\mathrm{rob}}_p$. 
The reduced-cost expression \eqref{eq:reduced-cost} and the SPRC reformulation in \cref{prop:sprc} remain valid after this substitution, with edge resource~$T(e)$ replaced by~$\tilde{T}^{\mathrm{rob}}(e)$ throughout.
\end{proof}

\begin{remark}[Path-level travel-time uncertainty beyond link-induced boxes]
\label{rem:path-uncertainty}
The analysis above assumes that path travel time uncertainty is induced from link-level box uncertainty as in \cref{def:unc-link-time}.
In this case, the path-wise inflation radii~$\rho^p_p$ decompose additively over the edges of a path, and the pricing problem remains an elementary \gls{acr:sprc} with edge-additive costs and resources.

If one instead postulates path-level travel-time uncertainty that is not induced from link times, the inflation terms~$\rho^p$ need not decompose additively over the edges of~$p$.
The reduced cost in the pricing problem then contains an additional path-level term, which breaks pure edge-additivity and complicates the dominance structure of the label-correcting algorithm.
To solve the robust counterpart, we can define a maximum difference in uncertainty by $\rho_{\Delta}=\max(\boldsymbol{\Delta^p})-\min(\boldsymbol{\Delta^p})$. In this case, the pricing represents a shortest path problem as before, with the dominance relationship in \cref{def:label} redefined as:
        \[c_1\leq c_2, t_1 -t_2 \leq \rho_{\Delta},\text{ and } \mathcal{V}^{\mathsf{vis}}_1 \subseteq \mathcal{V}^{\mathsf{vis}}_2.\]
The rest of the algorithm stays the same.
The robust counterpart exhibits higher complexity compared to the deterministic case, as the redefined dominance relationship results in weaker dominance pruning in the algorithm outlined in \cref{algo:rcsp}. This increased complexity is not a flaw of the proposed algorithm itself, but rather a characteristic of robust optimization in general, which tends to involve greater complexity compared to its deterministic. 
A detailed analysis of such general path-level uncertainty is beyond the scope of this paper; here we restrict attention to the structurally natural and practically relevant case of link-induced box uncertainty in \cref{def:unc-link-time}.
\end{remark}

\subsubsection{Uncertainty in demand}
We now consider box uncertainty in \gls{acr:od} demand. 
In the path-based formulation (\cref{form:path}), demand constraints read
\begin{equation}
\label{eq:det-demand}
\sum_{p\in P_{od}^{\mathrm{adm}}} f^p_{od}
\le \alpha_{od},
\qquad \forall (o,d)\in \tilde D.
\end{equation}
Let the nominal demand be~$\bar{\alpha}_{od}$ and model uncertain demand as
\[
\tilde{\alpha}_{od}
= \bar{\alpha}_{od} + \Delta^{\alpha}_{od},
\qquad |\Delta^{\alpha}_{od}| \le \rho^{\alpha}_{od}.
\]
Robust feasibility requires that \eqref{eq:det-demand} hold
for all~$\tilde{\alpha}_{od}$ in the box, i.e.,
\[
\sum_{p\in P_{od}^{\mathrm{adm}}} f^p_{od}
\le \tilde{\alpha}_{od}
\quad \forall \Delta^{\alpha}_{od}
\text{ with }
|\Delta^{\alpha}_{od}|\le \rho^{\alpha}_{od}.
\]
Since the right-hand side is uncertain and we are imposing
a ``$\le$'' constraint, the most restrictive realization is the smallest possible demand:
\begin{align}
\sum_{p\in P_{od}^{\mathrm{adm}}} f^p_{od}
&\le \min_{|\Delta^{\alpha}_{od}|\le \rho^{\alpha}_{od}}
(\bar{\alpha}_{od} + \Delta^{\alpha}_{od}) \nonumber\\
&= \bar{\alpha}_{od} - \rho^{\alpha}_{od}.
\label{eq:ro-demand}
\end{align}

\begin{lemma}[Compatibility under demand uncertainty]
\label{prop:ro-demand}
Under box uncertainty in \gls{acr:od} demand, the robust counterpart of the demand constraints~\eqref{eq:det-demand} is obtained by replacing each demand bound~$\alpha_{od}$ by~$\bar{\alpha}_{od} - \rho^{\alpha}_{od}$. 
The resulting robust \gls{acr:andp} has the same structure as the deterministic formulation, and the column-generation algorithm applies without modification.
\end{lemma}

\begin{proof}
Equation~\eqref{eq:ro-demand} shows that robust feasibility
is equivalent to tightening the right-hand side of the demand constraints. 
This change affects only the parameters~$\alpha_{od}$ in the master problem and the constant term in the dual objective; it does not alter the constraints of the dual
or the pricing problem. 
The decomposition and pricing structure are therefore unchanged.
\end{proof}

\subsubsection{Summary of robust extension}
Under box uncertainty, both in link travel times and in OD demand, the robust \gls{acr:andp} is obtained by:
\begin{itemize}
    \item inflating edge travel times from~$T(e)$ to
    $\tilde{T}^{\mathrm{rob}}(e) = \bar{T}(e) + \rho^T_e$ in
    the fleet-time constraint and in the pricing problem, and
    \item tightening OD demand bounds from~$\alpha_{od}$ to~$\bar{\alpha}_{od} - \rho^{\alpha}_{od}$.
\end{itemize}

These modifications preserve the structure exploited by the
proposed algorithm: the master problem remains an LP of the
same form, and the pricing subproblem remains an elementary
\gls{acr:sprc} on the road network with appropriately updated edge travel times. 
The robust and deterministic problems can therefore be solved by the \emph{same} \gls{acr:cg}-based framework, differing only in parameter values.

\subsection{Overall CG Algorithm and Optimality Gap}
\label{subsec:cg-alg}

We are now ready to present the full \gls{acr:cg} algorithm
for the (deterministic or robust) path-based \gls{acr:andp}. Starting from an initial restricted path set~$P^0 \subseteq P^{\mathrm{adm}}$, we iterate between solving the restricted LP master problem and solving a pricing problem for each OD pair in parallel using the SPRC algorithm of \cref{subsec:pricing-alg}. 
Whenever the maximum reduced cost is strictly positive, the related path is added to the path set.
The process terminates when no such path remains.

\begin{algorithm}[tb]
\small
\caption{Column generation for path-based formulation}
\label{algo:cg}
\begin{algorithmic}[1]
\Require Base network $\mathcal{G}=(\mathcal{V},\mathcal{E})$, OD set $\tilde{D}$, initial path $P^0 \subset P^{\mathrm{adm}}$; tolerance $\varepsilon>0$
\Ensure Final column set $P^\star$, LP solution $(\des_{\mathrm{LP}},\mathbf{f}_{\mathrm{LP}})$, and integer solution $(\des_{\mathrm{IP}},\mathbf{f}_{\mathrm{IP}})$ on $\mathcal{P}^\star$

\State $P^r\gets P^0$

\While{\textbf{true}}
  \State \textcolor{blue}{// Solve restricted master and read duals}
  \State Solve RMP restricted to $P^r$ and obtain $(\des,\mathbf{f})$
  \State Obtain dual multipliers $(\mu,\mathbf{v},\mathbf{u},\pi,\boldsymbol{\delta})$

  \State $P_{\text{new}}\gets \emptyset$

  \For{each $(o,d)\in\tilde{D}$ \textbf{in parallel}}
    \State \textcolor{blue}{// Pricing for OD: maximize reduced cost}
    \State Run \text{\Cref{algo:rcsp}}, obtain $p^*$ and reduced cost $\text{rc}(p^\star)$
    \If{$\text{rc}(p^\star)>\varepsilon$}
      \State $P_{\text{new}}\gets P_{\text{new}}\cup\{p^\star\}$
    \EndIf
  \EndFor

  \If{$P_{\text{new}}=\emptyset$}
    \State \textbf{break} \Comment{No violated dual constraints}
  \Else
    \State $P^r \gets P^r \cup P_{\text{new}}$
  \EndIf
\EndWhile

\State $P^\star \gets P^r$
\State $(\des_{\mathrm{LP}},\mathbf{f}_{\mathrm{LP}})\gets (\mathbf{x},\mathbf{f})$ \Comment{Optimal for LP relaxation over $P^\star$}

\State \textcolor{blue}{// Re-impose integrality and solve restricted MILP}
\State Solve \Cref{form:path} restricted to $P^\star$ and obtain $(\des_{\mathrm{IP}},\mathbf{f}_{\mathrm{IP}})$

\State \textbf{return } $P^\star$, $(\des_{\mathrm{LP}},\mathbf{f}_{\mathrm{LP}})$, $(\des_{\mathrm{IP}},\mathbf{f}_{\mathrm{IP}})$
\end{algorithmic}
\end{algorithm}

The solution algorithm can be summarized in three
conceptual steps:
\begin{enumerate}
    \item \textbf{Problem relaxation}: relax the
    \gls{acr:milp} \cref{form:path} to its
    \gls{acr:lp} relaxation (the master problem, \cref{prob:mp}) by allowing~$\des_{ij}\in[0,1]$.
    \item \textbf{\gls{acr:cg} iteration}: decompose the
    relaxed problem into a master problem and a pricing
    problem. 
    Solve them iteratively: the restricted master
    problem is solved over the current path set, and the
    pricing problem (an elementary \gls{acr:sprc}) is solved
    for each OD pair to generate new paths with positive
    reduced cost. 
    This yields a path set~$P^\star$ and an
    optimal solution of the LP relaxation over all admissible
    paths.
    \item \textbf{Integer solution recovery}: re-impose the
    integrality constraints on~$\des_{ij}$ and solve the resulting \gls{acr:milp} restricted to the path set $P^\star$.
\end{enumerate}

The complete algorithm is summarized in \ref{algo:cg}. Below are several notes on its implementation and properties. 

First, a tolerance parameter $\varepsilon$ is introduced to terminate the \gls{acr:cg} procedure when the maximum reduced cost falls below the threshold. This is necessary to account for numerical inaccuracies in \gls{acr:lp} solvers, which may return small positive values for reduced costs that are theoretically zero. Additionally, increasing $\varepsilon$ allows early termination of the process with limited computational resources. 

Second, in each iteration of the algorithm, the pricing problem is solved independently for each $(o,d) \in \tilde D$, and all paths with positive reduced cost above the tolerance are added to the restricted path set. This strategy can significantly reduce the number of \gls{acr:cg} iterations required for convergence. Alternative variants are possible, such as adding only the top $K$ paths with the largest positive reduced costs. The number of paths introduced per iteration reflects a tradeoff between the cost of re-solving the \gls{acr:rmp} and the total number of iterations until convergence. There is no closed-form rule for selecting this number, and it depends on the parameters of the problem.

Third, the algorithm does not, in general, guarantee global optimality for the original \gls{acr:milp}; it optimizes the LP relaxation exactly and then finds the best integer solution on the generated columns.
However, it provides an explicit, computable bound on the optimality gap.

Let~$J^{\star}_{\mathrm{LP}}$ denote the optimal objective value of the LP relaxation over the full path set~$P^{\mathrm{adm}}$, and let~$J^{\star}_{\mathrm{IP}}$ denote the
objective value of the integer solution obtained in Step~3 by solving the restricted MILP over~$P^\star$.

\begin{theorem}[Correctness and optimality gap]
\label{thm:cg-gap}
The \gls{acr:cg} algorithm
(\cref{algo:cg}) has the following properties:
\begin{enumerate}
    \item It terminates in finitely many iterations and returns
    an LP solution with objective value~$J^{\star}_{\mathrm{LP}}$ for the (deterministic or robust) master problem over the full admissible path set.
    \item The integer solution~$\tup{\des_{\mathrm{IP}},\mathbf{f}_{\mathrm{IP}}}$
    computed on~$P^\star$ is feasible for the original
    MILP (deterministic or robust), and the optimality gap $\epsilon$ is upper-bounded by the relative optimality gap bound
    \[
    \epsilon
    \leq \frac{J^{\star}_{\mathrm{LP}} -
             J^{\star}_{\mathrm{IP}}}{J^{\star}_{\mathrm{LP}}}.
    \]
\end{enumerate}
\end{theorem}

\begin{proof}
The proof follows the standard arguments for \gls{acr:cg}. By \cref{lem:label-optimality}, \cref{algo:rcsp} solves each pricing problem (SPRC) exactly, both in the deterministic and in the robust case (the latter after inflating edge travel times as in \cref{prop:ro-time}). 
When \cref{algo:cg} terminates, no path with positive
reduced cost exists for any OD pair, which implies that all dual constraints in the dual master problem are satisfied and no improving column remains. 
The current restricted master solution is therefore optimal for the LP relaxation over the \emph{full} admissible path set, with objective value~$J^{\star}_{\mathrm{LP}}$.

The LP relaxation provides an upper bound on the mixed‑integer optimum (deterministic or robust), so any feasible integer solution has objective value at most~$J^{\star}_{\mathrm{LP}}$.
The integer solution~$(\des_{\mathrm{IP}},\mathbf{f}_{\mathrm{IP}})$ constructed in Step~3 is feasible for the original MILP because it uses a subset~$P^\star \subseteq P^{\mathrm{adm}}$ of admissible paths and satisfies all constraints. 
Thus~$J^{\star}_{\mathrm{LP}} - J^{\star}_{\mathrm{IP}}$ bounds the absolute optimality gap, and dividing by~$J^{\star}_{\mathrm{LP}}$ yields the claimed relative gap.
\end{proof}

In large-scale instances, the generated path set~$P^\star$ is typically orders of magnitude smaller than the full admissible set~$P^{\mathrm{adm}}$, making the final restricted \gls{acr:milp} tractable while retaining a rigorous upper bound on the true optimum.
The algorithm exhibits an anytime property: additional \gls{acr:cg} iterations (or a tighter reduced‑cost
tolerance) monotonically improve~$J^{\star}_{\mathrm{LP}}$ and typically also improve the quality of the recovered integer solution. 
The same guarantees hold for both the deterministic and the robust \gls{acr:andp} under box uncertainty in travel times and demand.

\section{Case Study} \label{sec:exp}
In this section, we evaluate the proposed formulation and algorithm using real-world data from Manhattan, New York City. 
We start by describing the datasets and experimental setup.
We then use the framework to (i) characterize the structure and temporal stability of optimal operation subnetworks, (ii) study the joint impact of fleet-time and infrastructure budgets on profitability, and (iii) assess the effect of a path-level risk constraint that limits left turns. All experiments were conducted on a HPC cluster with CPU-only nodes (96 cores, ~520 GB RAM per node) with Gurobi 10.0.2. Each solve used 4 CPU threads and 16 GB RAM.

\subsection{Datasets}
The case study requires two primary inputs: a demand model and a base road network. 

\paragraph*{Demand data}
The demand data is derived from the NYC Taxi and Limousine Commission (TLC) dataset~\cite{nyc_tlc_trip_data}, which contains detailed trip-level records and is widely used in \gls{acr:amod} fleet control research~\cite{li2025reproducibility}.
We focus on trips whose origins and destinations lie within Manhattan and consider the month of May 2024, which corresponds to an average of approximately 200,000 trips per day.
Unless otherwise stated, we treat each day as an independent demand realization over the planning horizon and solve one instance of the \gls{acr:andp} per day.

\paragraph*{Road network}
The base road network, including the network structure, road length, expected travel time, and capacity, is constructed from OpenStreetMap (OSM)~\cite{OpenStreetMap}. 
The raw OSM representation of Manhattan consists of approximately 4,600 nodes and 9,900 directed edges. 
Because this representation is designed for geographic completeness rather than optimization, it includes many intermediate nodes and minor segments that are not necessary for network design.
We therefore preprocess the network by contracting degree-2 nodes on straight segments and retaining only arterial and local nodes that are essential for preserving connectivity.
After preprocessing, the network is reduced to approximately 4,300 nodes and 7,000 directed edges.

\paragraph*{Demand-network alignment}
In the TLC dataset, Manhattan is partitioned into 63 zones, with trip origins and destinations recorded at the zone level. 
To align demand with the road network, we project each zone centroid to the nearest node in the base network and use these nodes as \gls{acr:od} points. 
This yields a consistent representation in which each \gls{acr:od} pair in~$\tilde D$ is associated with a node pair in the road graph.

\begin{figure*}[tb]
    \centering
    \includegraphics[width=0.8\linewidth]{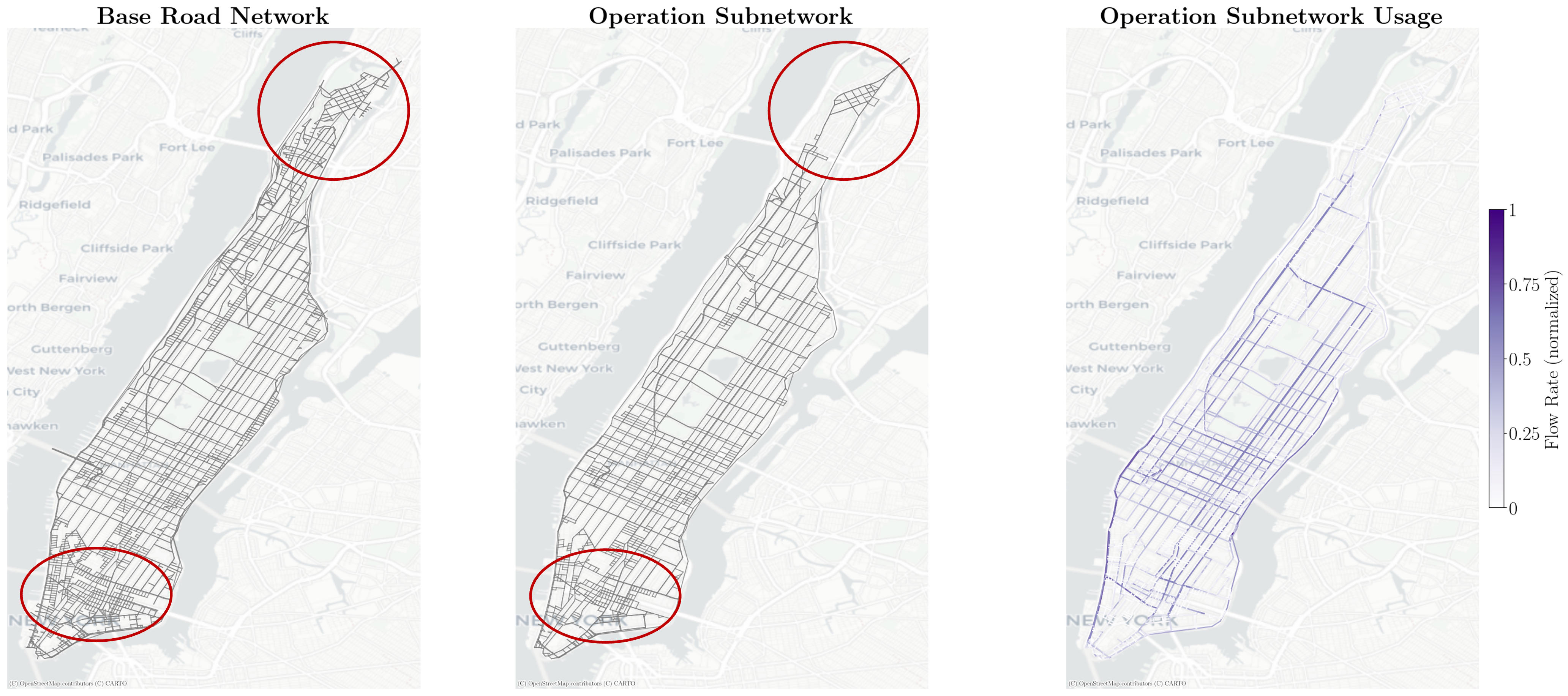}
    \caption{Example of solution under a fixed fleet-time and infrastructure budget.
    From left to right: base road network, optimal operation subnetwork, and
    optimal vehicle flow distribution on the operation subnetwork.}
    \label{fig:solution}
\end{figure*}

\subsection{Computational efficiency}
Before analyzing the structure of the optimal solution, we report the computational performance of the proposed framework. As presented in \cref{subsec:cg-alg}, the method solves the \gls{acr:lp} relaxation via a \gls{acr:cg}-based decomposition and then enforces integrality on the edge-instrumentation decisions $\des$ in a final \gls{acr:milp} solve. We run the algorithm for the 31 instances of demand from the month of May with a fleet-time budget of 16,000 vehicle hours and an infrastructure budget of \$600,000 US dollars. For all instances, the final \gls{acr:milp} solve terminated with a relative optimality gap between the optimal LP-relaxation objective and the optimal integer objective below $10^{-6}$ within 25 minutes.

The results indicate that the decomposition based on the \gls{acr:cg} significantly reduces the number of path variables in the final \gls{acr:milp}, from an unsolvable exponential number of variables to a tractable set. However, \gls{acr:milp} problems are generally challenging to solve, and their difficulty is highly instance-dependent. Therefore, the resulting formulations may have weaker relaxations and require substantially longer branch-and-bound search. Since accelerating worst-case \gls{acr:milp} solution times is not the focus of this work, we do not claim that all instances will be solved to proven optimality within a fixed runtime. For broader experimentation, a pragmatic approach is to impose a time limit and/or a prescribed relative optimality gap for the final \gls{acr:milp} solve, trading off compute time against solution quality. This is standard in large-scale applied \gls{acr:milp}: modern solvers expose relative/absolute MIP-gap tolerances explicitly as termination criteria, and setting such a gap is a commonly recommended option when proven optimality is computationally expensive~\cite{gurobiMIPGapGuidelines}.

\subsection{Operation subnetwork structure and temporal stability} \label{sec:solution}

We start analyzing the solutions by examining the structure of an optimal operation subnetwork and the associated vehicle flows.
Using demand data from May 15, 2024, \cref{fig:solution} shows (from left to right) the base road network, the optimal operation subnetwork, and the corresponding optimal vehicle flow distribution for a fleet-time budget of 16,000 vehicle hours and an infrastructure budget of \$600,000 US dollars.

A few observations follow.
First, as highlighted by the circled regions, not all local streets in the base network are chosen for autonomous operation. 
Instead, the optimal subnetwork focuses infrastructure on corridors that are most useful for serving the observed demand under the given resource constraints.
This suggests that, even in a dense urban grid, a relatively sparse operation network can achieve near-optimal performance when designed jointly with fleet capacity.
Second, the flow distribution on the operation subnetwork captures the spatial pattern of vehicle activity.
Edges with high flow in \cref{fig:solution} correspond to critical corridors and carry the bulk of passenger traffic.
Identifying these corridors is useful for a range of system-level decisions, such as depot placement, charging infrastructure, curb-space allocation, or targeted capacity upgrades.

The example in \cref{fig:solution} uses a single day of demand.
To assess how sensitive the optimal subnetwork is to day-to-day demand variability, we repeat the design computation for each day of May 2024, yielding 31 independent solutions under the same fleet-time and budget constraints.
\cref{fig:month}(a) reports, for each edge, the frequency with which it is included in the optimal operation subnetwork across these 31 days. 
Among all edges that appear at least once in the solutions, more than 90\% are included in over 25 of the 31 solutions, indicating that the optimal operation subnetwork is highly stable across days.
This stability reflects the regularity of urban demand patterns and suggests that an operator can design a medium-term operation network that remains near-optimal over many days.
From an operational perspective, this temporal stability also suggests that the subnetworks designed under our static fleet-time approximation are not artifacts of one specific demand realization.
Even though we do not model rebalancing explicitly, the induced operation subnetwork remains essentially unchanged across 31 independent daily demand realizations, which is consistent with the idea that infrastructure decisions evolve on a slower time scale than short-term fleet imbalances.

Finally, we illustrate the robust extension of \cref{subsec:robust}.
Using the one-month demand data, we instantiate a robust \gls{acr:andp} in which \gls{acr:od} demands are subject to box uncertainty calibrated from the observed variability.
\cref{fig:month}(b) shows the resulting robust operation subnetwork.
Compared to the edge-frequency map in \cref{fig:month}(a), the robust design highlights a core set of corridors that are consistently useful across demand realizations, providing a single network that hedges against day-to-day demand fluctuations.

\begin{figure}[tb]
    \centering
    \includegraphics[width=0.9\linewidth]{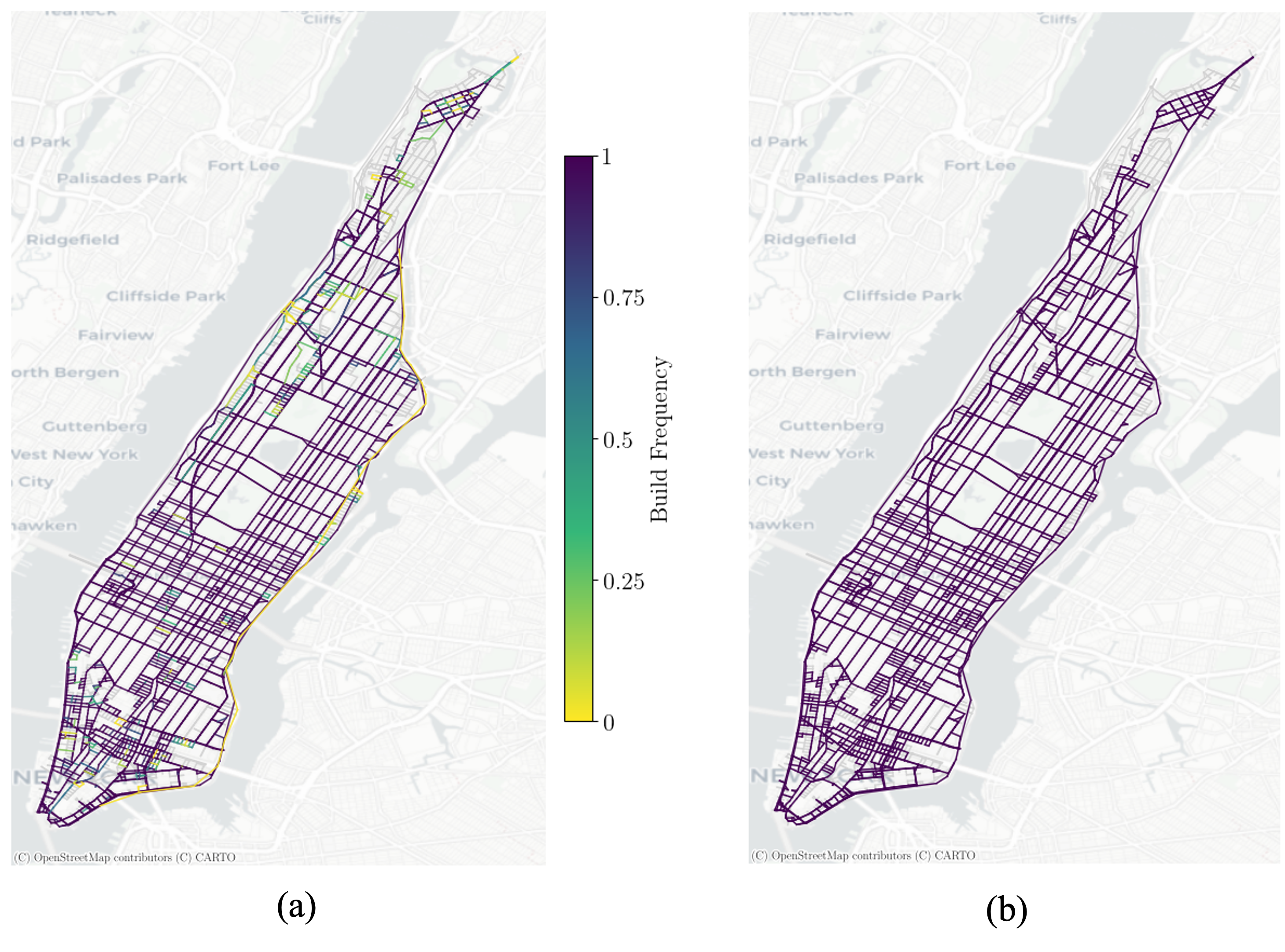}
    \caption{Design solutions with one month of demand data.
    (a) Edge instrumentation frequency over 31 daily designs.
    (b) Robust operation subnetwork obtained from the robust \gls{acr:andp}
    under demand uncertainty.}
    \label{fig:month}
\end{figure}

\subsection{Sensitivity to fleet time and infrastructure budget}
The previous subsection fixed the fleet-time and budget constraints.
We now study how service profitability responds to changes in these two design levers and how they interact.

Using demand data from May 15, we first solve the \gls{acr:andp} \emph{without} fleet-time or budget constraints, obtaining an unconstrained optimal profit~$F_b$ together with the corresponding fleet-time usage~$T_b$ and infrastructure cost~$C_b$.
We then resolve the \gls{acr:andp} for a grid of fleet-time limits and budget limits.
For ease of interpretation, we normalize profit, fleet time, and budget by~$F_b$,~$T_b$, and~$C_b$, respectively.
The resulting profit surface is shown in \cref{fig:sens}.
Three distinct regimes emerge:
\begin{enumerate}
    \item \emph{Fleet-limited regime:} when the fleet-time budget
    is small, increasing the infrastructure budget yields little improvement in profit, as vehicle availability is the binding constraint.
    \item \emph{Budget-limited regime:} when the infrastructure budget is small, increasing fleet time alone does not significantly improve performance because insufficient road instrumentation restricts feasible operations.
    \item \emph{Capacity-limited regime:} when both fleet time and budget are large, road capacity becomes the dominant constraint and profitability saturates, being ultimately    limited by the physical capacity of the network.
\end{enumerate}

The first two regimes underscore the importance of jointly designing fleet size and infrastructure investment.
Investing in only one resource leads to diminishing returns once the other becomes binding.
The third regime has direct policy implications: it illustrates how municipal regulations that effectively cap the number of \glspl{acr:av} on the road (e.g., through fleet caps or
access restrictions) can bound the maximum achievable profit.
In this sense, the proposed framework can also serve as a policy analysis tool, enabling municipalities to quantify the system-level impact of regulatory decisions.

Importantly, \cref{fig:sens} shows that the structure of the optimal operation subnetwork changes smoothly across wide ranges of~$R$ and~$B$: there is no evidence of highly fragile designs that would collapse when the available fleet time is perturben within a realistic range.
This provides an indirect robustness check for the fleet-time proxy.
While detailed dynamic rebalancing may shift the level of the required budget, the relative trade-offs and the selected corridors remain stable across the tested operating regimes.

\begin{figure}[tb]
    \centering
    \includegraphics[width=0.9\linewidth]{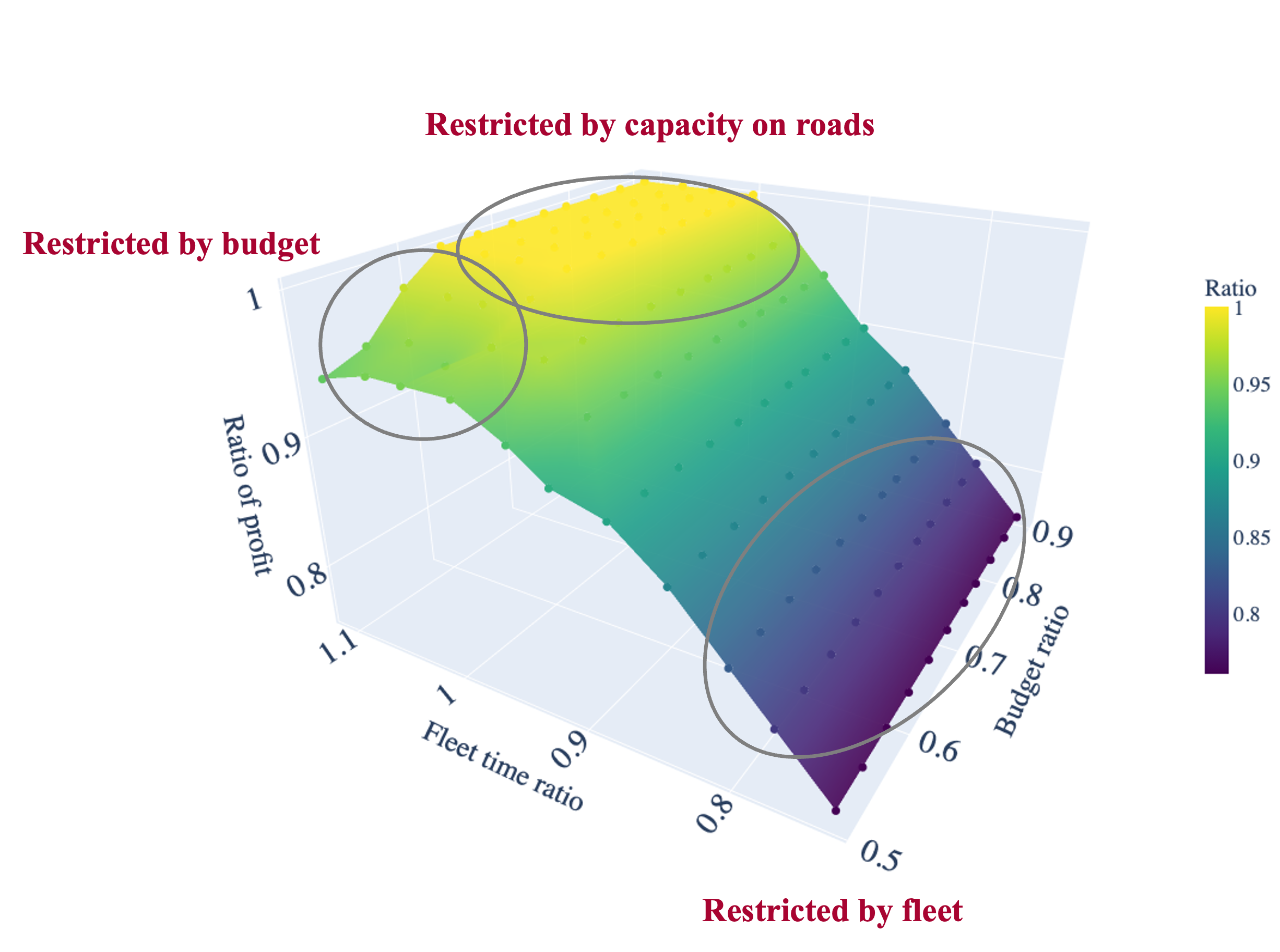}
    \caption{Sensitivity of normalized profit to fleet-time and infrastructure budget limits.
    The circled regions highlight regimes in which performance is limited by
    different constraints.}
    \label{fig:sens}
\end{figure}

\subsection{Path-level risk constraint: limiting left turns} \label{sec:left-turn-constraint}
We conclude the case study by illustrating how the framework handles additional path-level constraints, focusing on a risk-related constraint that limits left turns.

In autonomous driving, left-turn maneuvers, especially unprotected left turns, are a prominent intersection crash mode and pose significant challenges due to occlusions, complex interactions, and conflicting traffic streams. 
They are often used as canonical test scenarios for \gls{acr:av} decision-making~\cite{waymo2021safety}.
Limiting the frequency of such maneuvers has the potential to reduce operational risk.

To capture this effect in our model, we introduce a constraint on the total number of left turns performed by the fleet during operation.
Since OSM does not provide signal-phase information, we conservatively treat all left turns as unprotected and risky, and augment \cref{form:path} with the constraint
\begin{equation}
    \sum_{(o,d)\in\tilde{D}} \sum_{p \in P_{od}} lt_p f_{od}^{p} \le LT,
    \label{eq:lt_constraint}
\end{equation}
where~$lt_p$ is the number of left turns on path~$p$ and~$LT$ is the total allowed number of left turns over the planning horizon.
The parameter~$LT$ can be interpreted as a system-wide risk budget: smaller values enforce safer routing at the cost of reduced flexibility.

Let $\omega^\star$ denote the optimal dual variable associated with \cref{eq:lt_constraint}.
The pricing objective then becomes:
\[\beta^p_{od} - \sum_{(i,j) \in \mathcal{E}}u^*_{ij}z^{od,p}_{ij} - v^*_{od} - \mu^* t_p - \omega^\star lt_p,\]
which is equivalent to an \gls{acr:sprc} with an additional left-turn cost $\omega^\star$. Solving this problem requires only minor modification to \Cref{algo:rcsp}: the pricing algorithm keeps track of the predecessor node of each label to detect left turns and adds a cost $\omega^\star$ whenever a left turn is taken.
All other components of the formulation and algorithm remain
unchanged.

Using the same fleet time and budget constraints as in \Cref{sec:solution}, we solve the \gls{acr:andp} with and without the left-turn constraint for each day of May 2024 and for a range of~$LT$ values.
\Cref{fig:lt} reports, across the 31 days, the ratio of served demand and profit between the constrained and unconstrained solutions. 

With a tight left-turn budget (small $LT$), the number of admissible paths is severely reduced, limiting the system’s ability to serve demand and resulting in significant losses in served demand and profit.
As $LT$ increases, both metrics improve and eventually stabilize, indicating a regime in which additional left turns provide little marginal benefit.
\begin{figure}[tb]
    \centering
    \includegraphics[width=1.0\linewidth]{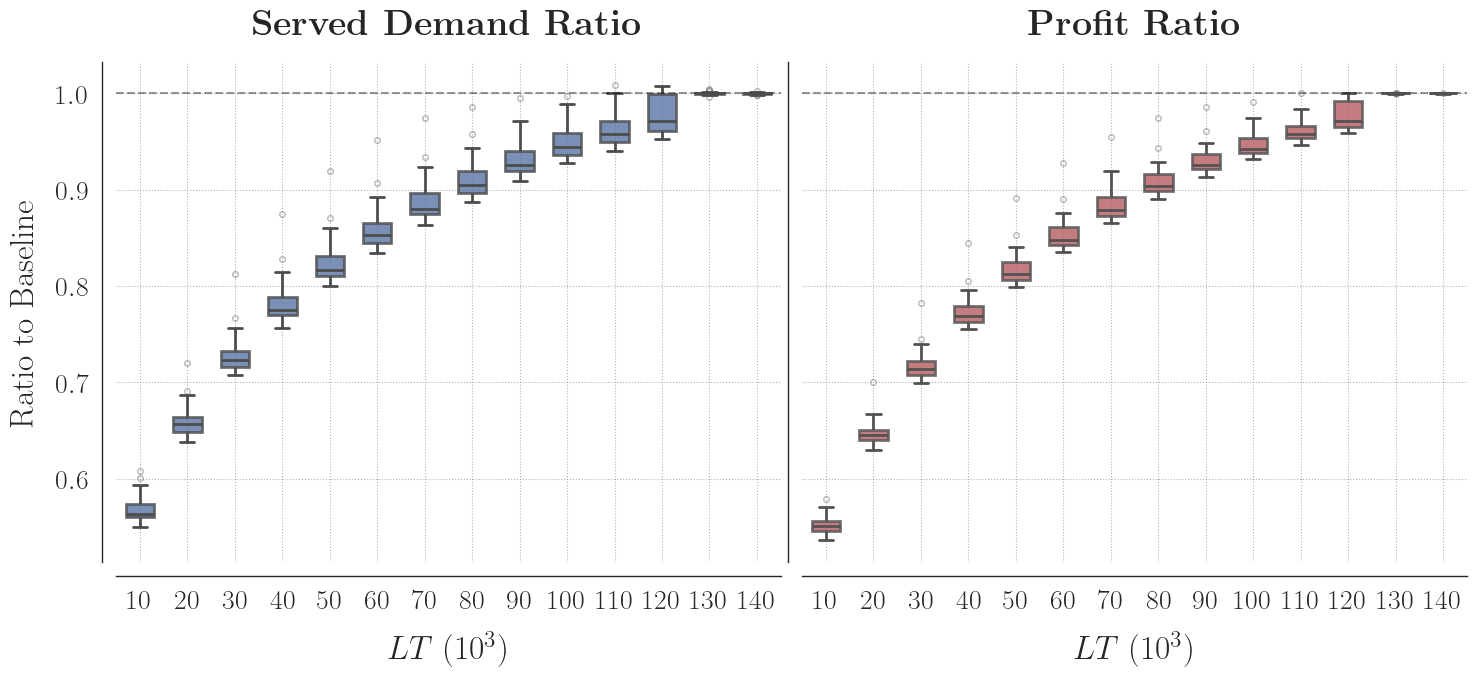}
    \caption{Ratio of served demand (left) and profit (right) with and without
    the left-turn constraint, for different values of the left-turn budget $LT$.
    The box plots show variability across the 31 days of May 2024.}
    \label{fig:lt}
\end{figure}

These results illustrate the substantial impact that path-level risk constraints can have on system performance and show how the proposed framework can be used to explore safety–performance trade-offs.
For operators, the left-turn budget $LT$ provides a tunable parameter to balance safety and profitability.
For municipalities, such experiments can inform guidelines on acceptable risk levels and help assess the impact of intersection design or turn restrictions on emerging \gls{acr:amod} services.

\section{Conclusion and future work} \label{sec:conclusion}
This paper studied the strategic design problem for \gls{acr:amod} systems, where a centralized operator must decide both \emph{where} autonomous vehicles can operate and \emph{how} to use a limited fleet to serve demand. 
We formalized this problem as the \gls{acr:andp}, in which the operator selects an operation subnetwork, chooses a fleet-time budget, and routes all passengers subject to infrastructure and fleet constraints and route-level \gls{acr:qos} and risk requirements. 
To address the resulting high-dimensional mixed-integer problem, we proposed a path-based formulation and a \gls{acr:cg}-based decomposition algorithm. The master problem solves the LP relaxation over a restricted path set, while the pricing problem reduces to an
elementary \gls{acr:sprc} solved exactly by a tailored label-correcting algorithm. 
The approach provides the exact LP optimum and an explicit certificate on the optimality gap of the recovered integer solution, and it extends to a robust counterpart under box uncertainty in travel times and demand with only minor parameter changes.

Using real-world network and demand data from Manhattan,
New York City, we demonstrated that the framework scales to city-sized instances and delivers high-quality solutions within practical computation times. 
The case study showed that the optimal operation subnetworks are structurally sparse yet highly stable across days, highlighted how profitability is jointly shaped by infrastructure budgets and fleet time limits, and illustrated the impact of path-level risk constraints, such as limits on left turns, on both safety and performance. 
These results indicate that the proposed framework can serve as a decision-support tool for operators and municipalities when
planning future \gls{acr:amod} deployments.

Several open questions remain and point to promising
directions for future research.
First, we plan to extend the robust extension beyond box uncertainty to richer uncertainty sets and alternative risk measures, enabling more nuanced views of demand and travel-time variability. 
Second, our steady‑state formulation abstracts away explicit empty‑vehicle rebalancing and time‑of‑day dynamics. 
While the Manhattan case study suggests that the resulting infrastructure designs are temporally stable, a more complete picture would couple network design with a dynamic \gls{acr:amod} control layer that models rebalancing flows explicitly, for example via a time‑expanded network or fluid model. 
We view such a multi‑scale design‑and‑control framework as a key direction for future research.
Third, we see potential in integrating the \gls{acr:andp} with public transit and other modes, co-optimizing multimodal services~\cite{zardini2022co}.

\bibliographystyle{IEEEtran}
{\footnotesize
\bibliography{references}}

\end{document}